\documentclass{article}
\usepackage[a4paper,total={6in, 8in}]{geometry}

\usepackage{amsmath,amssymb,amsthm,wasysym,epsfig,color,subfigure,graphicx,epstopdf}
\usepackage{amsmath,amssymb,amsthm,wasysym,epsfig,color,subfigure,graphicx,epstopdf,url,bm,nicefrac}
\usepackage[table]{xcolor}
\usepackage{bbold}
\usepackage{caption}

\usepackage{algorithm}
\usepackage{algorithmic}

\allowdisplaybreaks

\usepackage{hyperref}
\hypersetup{colorlinks=false}

\usepackage{accents}
\makeatletter
\def\wid{\check{{\cc@style\underline{\mskip9.5mu}}}}
\def\Wideubar{\underaccent{{\cc@style\underline{\mskip6mu}}}}
\makeatother

\makeatletter
\def\wideubar{\underaccent{{\cc@style\underline{\mskip9.5mu}}}}
\def\Wideubar{\underaccent{{\cc@style\underline{\mskip6mu}}}}
\makeatother

\makeatletter
\def\widebar{\accentset{{\cc@style\underline{\mskip9.5mu}}}}
\def\Widebar{\accentset{{\cc@style\underline{\mskip6mu}}}}
\makeatother

\newtheorem{lemma}{Lemma}

\newtheorem{theorem}{Theorem}

\theoremstyle{remark}\newtheorem{remark}{Remark}

\newcommand{\minimize}{{\rm minimize}}

\def\ccalH{{\ensuremath{\mathcal H}}}
\def\ccalT{{\ensuremath{\mathcal T}}}

\begin{document}
\title{Sparse Phase Retrieval via Truncated\\ Amplitude Flow
}

\author{
	Gang Wang,
	Liang Zhang,
	Georgios B. Giannakis,\\
		Mehmet Ak\c{c}akaya, and 
	Jie Chen\thanks{	The work of G. Wang, L. Zhang, and G. B. Giannakis was supported in part by NSF grants 1500713 and 1514056. 
	G. Wang, L. Zhang, G. B. Giannakis, and M. Ak\c{c}akaya are with the Digital Technology Center and the Department of Electrical and Computer Engineering, University of Minnesota, Minneapolis, MN 55455, USA. 
		G. Wang is also with the State Key Laboratory of Intelligent Control and Decision of Complex Systems, Beijing Institute of Technology, Beijing 100081, P. R. China.
		J. Chen is with the School of Automation and State Key Laboratory of Intelligent Control and Decision of Complex Systems, Beijing Institute of Technology, Beijing 100081, P. R. China. E-mails: \{gangwang,\,zhan3523,\,georgios,\,akcakaya\}@umn.edu; chenjie@bit.edu.cn.}
}

\maketitle

\allowdisplaybreaks

\begin{abstract}

This paper develops a novel algorithm, termed \emph{SPARse Truncated Amplitude flow} (SPARTA), to reconstruct a sparse signal from a small number of magnitude-only measurements. It deals with what is also known as sparse phase retrieval (PR), which is \emph{NP-hard} in general and emerges in many science and engineering applications. Upon formulating sparse PR as an amplitude-based nonconvex optimization task, SPARTA works iteratively in two stages: In stage one, the support of the underlying sparse signal is recovered using an analytically well-justified rule, and subsequently a sparse orthogonality-promoting initialization is obtained via power iterations restricted on the support; 
and, in the second stage, the initialization is successively refined by means of hard thresholding based gradient-type iterations. SPARTA is a simple yet effective, scalable, and fast sparse PR solver. On the theoretical side, for any $n$-dimensional $k$-sparse ($k\ll n$) signal $\bm{x}$ with minimum (in modulus) nonzero entries on the order of $(1/\sqrt{k})\|\bm{x}\|_2$, SPARTA recovers the signal exactly (up to a global unimodular constant) from about $k^2\log n$ random Gaussian measurements with high probability. Furthermore, SPARTA incurs computational complexity on the order of $k^2n\log n$ with total runtime proportional to the time required to read the data, which improves upon the state-of-the-art by at least a factor of $k$. Finally, SPARTA is robust against additive noise of bounded support. Extensive numerical tests corroborate markedly improved recovery performance and speedups of SPARTA relative to existing alternatives. 

\end{abstract}


\section{Introduction}\label{sec:intro}

In many fields of engineering and applied physics, one is often tasked with reconstructing a signal from the (squared) modulus of its Fourier (or any linear) transform, which is also known as phase retrieval (PR). Such a task arises naturally in applications such as X-ray crystallography, microscopy and ptychography, astronomy, optics, as well as array and coherent diffraction imaging. In these settings, optical sensors and detectors such as charge-coupled device cameras, photosensitive films, and human eyes record only the intensity (squared magnitude) of a light wave, but not the phase. In particular, solution to PR has led to significant accomplishments, including the discovery in $1953$ of DNA double helical structure from diffraction patterns, and the characterization of aberrations in the Hubble Space Telescope from measured point spread functions~\cite{spm2016ehm}.
Due to the absence of Fourier phase information, the one-dimensional ($1$D) Fourier PR problem is generally ill-posed. It can be shown that there are in fact exponentially many non-equivalent solutions beyond trivial ambiguities in the $1$D PR case~\cite{1duniqueness}. A common approach to overcome this ill-posedness is exploiting additional information on the unknown signal such as non-negativity, sparsity, or bounded magnitude~\cite{gespar,1982fienup,tsp2015njs}. Other viable solutions consist of introducing redundancy into the measurement transforming system to obtain over-sampled and short-time Fourier transform (STFT) measurements~\cite{arxiv2016be}, random Gaussian measurements~\cite{wf,twf,taf}, and coded diffraction patterns using structured illumination and random masks~\cite{siam2015candes,coded,wf}, just to name a few; see~\cite{siam2015candes} for contemporary reviews on the theory and practice of PR. 

Past PR approaches can be mainly categorized as convex and nonconvex ones.
A popular class of nonconvex approaches is based on alternating projections including the seminal works by Gerchberg-Saxton~\cite{gerchberg} and Fienup~\cite{1982fienup},~\cite{2015chen1}, \cite{wald2}, alternating minimization with re-sampling (AltMinPhase)~\cite{tsp2015njs}, (stochastic) truncated amplitude flow (TAF)~\cite{nips2016wg,taf,staf,spl2016cl,raf} and the Wirtinger flow (WF) variants~\cite{wf,twf,reshaped1,tsp2017chi}, trust-region~\cite{sun2016}, (stochastic) proximal linear algorithms \cite{duchi2017,2017stochastic}. See also related discussion in~\cite{siam2015candes,spm2016ehm,2017mm,2017cr,variational2017chang,raf1,li2017sublinear,tsp2016qian}. Specifically, the WF variants and the trust-region methods minimize the intensity (modulus squared) based empirical risk, while AltMinPhase and TAF cope with the amplitude-based empirical risk.  
The convex alternatives either rely on the so-called Shor's relaxation to obtain semidefinite programming (SDP) based solvers 
abbreviated as PhaseLift~\cite{phaselift} and PhaseCut~\cite{phasecut}, or solve a basis pursuit problem in the dual domain as in PhaseMax~\cite{phasemax,phasemaxproof,2017dl}. 

Nevertheless, in various applications, especially those related to imaging, the underlying signal is naturally sparse or admits a sparse representation after some known and deterministic linear transformation~\cite{spm2016eldar}.  
For example, astronomical imaging centers around sparsely distributed stars, while electron microscopy deals with sparsely distributed atoms or molecules. As PR of sparse signals is of practical relevance, SDP, AltMinPhase, and WF recovery methods have been generalized to sparse PR producing
solvers termed compressive phase retrieval via lifting (CPRL)~\cite{cprl}, sparse AltMinPhase~\cite{tsp2015njs},  thresholded Wirtinger flow (TWF)~\cite{as2016clm}, SparsePhaseMax \cite{sparsephasemax}.
 CPRL in particular, accounts for the sparsity by adding an $\ell_1$-regularization term on the wanted signal to the original PhaseLift formulation. 
The other two approaches are two-stage iterative counterparts consisting of a (sparse) initialization, and a series of refinements of the initialization with gradient-type iterations. 
The greedy sparse phase retrieval (GESPAR) algorithm is based on a fast $2$-opt local search~\cite{gespar}. A probabilistic approach is developed based on the generalized approximate message passing (GAMP) algorithm~\cite{tsp2015sr}. Majorization-minimization algorithms are devised in \cite{tsp2016qp}.  
Assuming noise-free Gaussian random measurements, CPRL recovers any $k$-sparse $n$-dimensional ($k\ll n$) 
signal exactly from \footnote{The notation $\phi(n)=\mathcal{O}(g(n))$ means that there is a constant $c>0$ such that $|\phi(n)|\le c|g(n)|$.} $\mathcal{O}(k^2\log n )$ measurements at computational complexity $\mathcal{O}(n^3)$~\cite{sjam2013lv}. 
Sparse AltMinPhase and TWF, on the other hand, require $\mathcal{O}(k^2\log n)$ measurements \cite{tsp2015njs,as2016clm}, and SparseAltMinPhase incurs complexity  
$\mathcal{O}(k^2n\log n)$~\cite{tsp2015njs}.


Building on TWF and TAF, we propose here a novel sparse PR algorithm, which we call 
\emph{SPARse Truncated Amplitude flow} (SPARTA). Adopting an amplitude-based nonconvex formulation of the sparse PR, SPARTA emerges as a two-stage iterative solver: In stage one, the support of the underlying signal is estimated first using a well-justified rule, 
and subsequently power iterations are employed to obtain an initialization restricted on the recovered support; 
while the second stage successively refines the initialization with a series of
hard thresholding based truncated gradient iterations.
Both stages are conceptually simple, scalable, and fast. 
Moreover, we demonstrate that SPARTA recovers any $k$-sparse $n$-dimensional real-/complex-valued signal $\bm{x}$
($k\ll n$) with minimum nonzero entries (in modulus) on the order of $(1/\sqrt{k})\|\bm{x}\|_2$
from $\mathcal{O}(k^2\log n)$ measurements.
Further, to reach any given solution accuracy $\epsilon>0$, 
SPARTA incurs total computational cost of $\mathcal{O}(k^2n\log n \log(1/\epsilon))$, which improves upon the state-of-the-art by at least a factor of $k$. This computational advantage is paramount in large-scale imaging applications, where the basis factor $n\log n$ is large, 
typically on the order of millions.
In addition, SPARTA can be shown robust to additive noise of bounded support. Extensive simulated tests demonstrate markedly improved exact recovery performance (in the absence of noise), robustness to noise, and runtime speedups relative to the state-of-the-art algorithms.

The remainder of this paper is organized as follows. Section \ref{sec:prob} reviews the sparse PR problem, and also presents known necessary and sufficient conditions for uniqueness. Section~\ref{sec:alg} details the two stages of the proposed algorithm, whose analytic performance analysis is the
subject of Section~\ref{sec:main}. 
Finally, numerical tests are reported in~Section~\ref{sec:test}, proof details are given in~Section~\ref{sec:proof}, and conclusions are drawn in Section~\ref{sec:con}. Supporting lemmas are presented in the Appendix. 

Regarding common notation used throughout the paper, lower- (upper-) case boldface letters denote column vectors (matrices) of suitable dimensions, and symbol $\ccalT$ ($\ccalH$) as superscript stands for matrix/vector transposition (conjugate transposition). Calligraphic letters are reserved for sets, e.g., $\mathcal{S}$. For vectors, $\|\!\cdot\!\|_2$ represents the Euclidean norm, while $\|\!\cdot\!\|_0$ denotes the $\ell_0$ pseudo-norm counting the number of nonzero entries. Finally, the ceiling operation $\lceil\cdot\rceil$ returns the smallest integer greater than or equal to the given number, and the cardinality $|\mathcal{S}|$ reports the number of elements in the set $\mathcal{S}$.     

\section{Sparse Phase Retrieval}\vspace{-.em}\label{sec:prob}

Succinctly stated, the sparse PR task amounts to reconstructing a sparse $\bm{x}\in\mathbb{R}^n$ 
(or $\mathbb{C}^n$) 
given a system of phaseless quadratic equations taking the form~\cite{spie2007mrb}
\begin{equation}\label{eq:quad}
\psi_i=|\langle\bm{a}_i,\bm{x}\rangle|,\quad 1\le i \le m,\qquad\text{subject~to}\quad\|\bm{x}\|_0\le k
\end{equation}
where $\{\psi_i\}_{i=1}^m$ are the observed modulus data, and $\{\bm{a}_i\}_{i=1}^m$ are known sensing (feature) vectors. 
The sparsity level $ k\ll n$ is assumed known \emph{a priori} for theoretical analysis purposes, 
while numerical implementations with unknown $k$ values will be tested as well.  
Alternatively, the data can be given in modulus squared (i.e., intensity) form as $\{y_i=|\langle\bm{a}_i,\bm{x}\rangle|^2\}_{i=1}^m$. 
It has been established that $m=2k$ generic \footnote{It is not within the scope of this paper to explain the meaning of generic vectors. Interested readers are referred to~\cite{4m-4}.} (e.g., random Gaussian) measurements as in \eqref{eq:quad} are necessary and sufficient for uniquely determining a $k$-sparse solution in the real case, and $m\ge 4k-2$ are sufficient in the complex case~\cite{spl2015at}. In the noisy scenario, stable compressive PR requires at least as many measurements as the corresponding compressive sensing problem since one is tasked with even less (no phase) information. 
Hence, stable sparse PR requires at least  
$\mathcal{O}(k\log(n/k))$ measurements as in compressive sensing~\cite{acha2015ivw}. Indeed, it has been recently demonstrated that $\mathcal{O}(k\log(n/k))$ generic measurements also suffice for stable  PR of a real-valued sparse signal~\cite{acha2014yesm}. 

For concreteness of our analytical results, the present paper focuses on the real-valued Gaussian model, which assumes independently and identically distributed (i.i.d.) standard Gaussian sensing vectors $\bm{a}_i\sim\mathcal{N}(\bm{0},\,\bm{I}_n)$, $i=1,\,\ldots,\,m$, and $\bm{x}\in\mathbb{R}^n$. 
Nevertheless, our proposed algorithm works also for the complex-valued Gaussian model with $\bm{x}\in\mathbb{C}^n$ and i.i.d. $\bm{a}_i\sim\mathcal{CN}(\bm{0},\bm{I}_n):=\mathcal{N}(\bm{0},\bm{I}_n/2)+j\mathcal{N}(\bm{0},\bm{I}_n/2)$.
Given $\{(\bm{a}_i,\psi_i)\}_{i=1}^m$ and assuming also the existence of a unique $k$-sparse solution (up to a global sign), our objective is to develop simple yet effective algorithms to provably reconstruct any $k$-sparse $n$-dimensional signal $\bm{x}$ from a small number (far less than $n$) of phaseless quadratic equations as in \eqref{eq:quad}. 

Adopting the least-squares criterion (which coincides with the maximum likelihood one when assuming additive white Gaussian noise in \eqref{eq:quad}), the problem of recovering a $k$-sparse solution from phaseless quadratic equations naturally boils down to that of minimizing the ensuing amplitude-based empirical loss function\vspace{-.em}
\begin{equation}\vspace{-.em}
\label{eq:cost}
	\underset{\|\bm{z}\|_0= k}{\text{minimize}}~~\ell(\bm{z}):=\frac{1}{2m}\sum_{i=1}^m\left(\psi_i-|\bm{a}_i^\ccalT\bm{z}|\right)^2.
\end{equation}  

Clearly, 
both the objective function and the $\ell_0$-norm constraint in~\eqref{eq:cost}
are nonconvex, which render the optimization problem \emph{NP-hard} in general~\cite{nphard}, and thus computationally intractable. Besides nonconvexity, another notable challenge here involves the non-smoothness of the cost function. 
It is worth emphasizing that (thresholded) Wirtinger alternatives dealt with the smooth counterpart of \eqref{eq:cost} based on squared magnitudes $\{y_i=|\bm{a}_i^\ccalT\bm{z}|^2\}_{i=1}^m$, 
which was numerically and experimentally shown to be less effective than the amplitude-based one even when no sparsity is exploited~\cite{taf,experimental2015}. 
Although focusing on a formulation similar to (but different than)~\eqref{eq:cost}, 
sparse AltMinPhase first estimates the support of the underlying signal, and performs standard PR of signals with dimension $k$. More importantly, sparse AltMinPhase relying on alternating minimization with re-sampling entails solving a series of least-squares problems, and performs matrix inversion at every iteration. Numerical tests suggest that a very large number of measurements are 
required to estimate the support exactly. Once wrong, sparse AltMinPhase confining the PR task on the estimated support would be impossible to recover the underlying sparse signal. On the other hand, motivated by the iterative hard thresholding (IHT) algorithms for compressive sensing~\cite{acha2009bd,acha2009nt}, 
an adaptive hard thresholding procedure that maintains only certain largest entries per iteration during the gradient refinement stage turns out to be effective~\cite{as2016clm}.   
Yet both sparse AltMinPhase and TWF were based on the simple spectral initialization, which was recently shown to be less accurate and robust than the orthogonality-promoting initialization~\cite{taf}. 

Broadening the TAF approach and the sparse PR solver TWF, the present paper puts forth a novel iterative solver for~\eqref{eq:cost} that proceeds in two stages: S1) a sparse orthogonality-promoting initialization is obtained by solving a PCA-type problem with a few simple power iterations on an estimated support of the underlying sparse signal; 
and, S2) successive refinements of the initialization are effected by means of a series of truncated gradient iterations along with a hard thresholding per iteration to set all entries to zero, except for the $k$ ones of largest magnitudes. The two stages are presented in order next.

\section{Algorithm: Sparse Truncated Amplitude Flow}
\label{sec:alg}
In this section, the initialization stage and the gradient refinement stage of SPARTA will be described in detail.
  To begin, let us introduce the distance from any estimate $\bm{z}\in\mathbb{R}^n$ to the solution set $\{\pm\bm{x}\}\subseteq\mathbb{R}^n$ to be ${\rm dist}(\bm{z},\bm{x}):=\min\{\|\bm{z}+\bm{x}\|_2,\|\bm{z}-\bm{x}\|_2\}$.  
 Define also the indistinguishable global phase constant in the real case as 
 \begin{equation}\vspace{-.em}
\label{eq:sign}
 \phi(\bm{z}):=\left\{\begin{array}
 	{ll}
 	0,&~\|\bm{z}-\bm{x}\|_2\le \|\bm{z}+\bm{x}\|_2,\\
 	\pi,&~\text{otherwise}.
 \end{array}
 \right.
 \end{equation}
 Hereafter, assume $\bm{x}$ to be
   the fixed solution to problem~\eqref{eq:quad} with $\phi(\bm{z})=0$; otherwise, one can replace $\bm{z}$ by $\bm{z}\text{e}^{i\phi}$, but the constant phase shift shall be dropped for notational brevity.
 Assume also without loss of generality that $\|\bm{x}\|_2=1$, which will be justified and generalized shortly.  

\vspace{-.em}
\subsection{Sparse Orthogonality-promoting Initialization}\vspace{-.em}
\label{sub:init}
When no sparsity is exploited, 
the orthogonality-promoting initialization proposed in~\cite{taf} 
starts with a popular folklore in stochastic geometry: High-dimensional random vectors are almost always nearly orthogonal to each other~\cite{jmlr2013cai}. 
The key idea is approximating the unknown $\bm{x}$ by another vector that is most orthogonal to a carefully chosen subset of sensing vectors $\{\bm{a}_i\}_{i\in\mathcal{I}^0}$, where $\mathcal{I}^0\subseteq [m]:=\{1,\,2,\,\ldots,\,m\}$ is some index set to be designed next. 
It is well known that the orthogonality between two vectors can be interpreted by their squared normalized inner-product  $(\bm{a}_i^\ccalT\bm{x})^2/(\|\bm{a}_i\|_2^2\|\bm{x}\|_2^2)$. Intuitively, the smaller the squared normalized inner-product between two vectors $\bm{a}_i$ and $\bm{x}$ is, the more orthogonal they are to each other. Upon evaluating the inner-product between each $\bm{a}_i$ and $\bm{x}$ for all pairs $\{(\bm{a}_i,\bm{x})\}_{i=1}^m$, one can construct $\mathcal{I}^0$ to include the indices of $\{\bm{a}_i\}$'s corresponding to the $|\mathcal{I}^0|$-smallest squared normalized inner-products with $\bm{x}$. Therefore, it is natural to approximate $\bm{x}$ by computing a vector $\bm{z}^0$ most orthogonal to the set $\mathcal{I}^0$ of sensing vectors~\cite{taf}. Mathematically, 
this is equivalent to solving a smallest eigenvector (defined to be the eigenvector associated with the smallest eigenvalue of a symmetric positive definite matrix) problem
\begin{equation}
	\label{eq:mineig}
	\underset{\|\bm{z}\|_2=1}{\text{minimize}}
	~~\bm{z}^\ccalT\bm{Y}	\bm{z}:=\bm{z}^\ccalT
	\Big(\frac{1}{|\mathcal{I}^0|}\sum_{i\in\mathcal{I}^0}\frac{\bm{a}_i\bm{a}_i^\ccalT}{\|\bm{a}_i\|_2^2}\Big)
	\bm{z}.
\end{equation}

The smallest eigenvalue (eigenvector) problem can be solved by fully eigen-decomposing the matrix $\frac{1}{|\mathcal{I}^0|}\sum_{i\in\mathcal{I}^0}\frac{\bm{a}_i\bm{a}_i^\ccalT}{\|\bm{a}_i\|_2^2}$ at computational complexity $\mathcal{O}(n^3)$ (assuming $|{\mathcal{I}}^0|$ to be on the order of $n$). Upon defining $\widebar{\mathcal{I}}^0$ to be the complement of the set $\mathcal{I}^0$ in $[m]$, one can rewrite
$
\sum_{i\in\mathcal{I}^0}\frac{\bm{a}_i\bm{a}_i^\ccalT}{\|\bm{a}_i\|_2^2}=\sum_{i\in [m]}\frac{\bm{a}_i\bm{a}_i^\ccalT}{\|\bm{a}_i\|_2^2}
-\sum_{i\in\widebar{\mathcal{I}}^0}\frac{\bm{a}_i\bm{a}_i^\ccalT}{\|\bm{a}_i\|_2^2}.
$
Recall that for i.i.d. standard Gaussian sensing vectors $\{\bm{a}_i\sim\mathcal{N}(\bm{0},\bm{I}_n)\}_{i=1}^m$,  
the following concentration result holds~\cite{chap2010vershynin}
\begin{equation}
\label{eq:concen}
\frac{1}{m}\sum_{i=1}^m\frac{\bm{a}_i\bm{a}_i^\ccalT}{\|\bm{a}_i\|_2^2}\approx \mathbb{E}\Big[\frac{\bm{a}_i\bm{a}_i^\ccalT}{\|\bm{a}_i\|_2^2}\Big]=\frac{1}{n}\bm{I}_n
\end{equation}
where $\mathbb{E}[\cdot]$ denotes the expected value.
It follows from~\eqref{eq:concen} that
the smallest eigenvector problem in~\eqref{eq:mineig} can be approximated by the largest (principal) eigenvector 
\begin{equation}
\label{eq:maxeig}
	\tilde{\bm{z}}^0:=\arg	\underset{\|\bm{z}\|_2=1}{\text{max}}~~
\bm{z}^\ccalT\widebar{\bm{Y}}\bm{z}
:=\bm{z}^\ccalT\Big(\frac{1}{|\widebar{\mathcal{I}}^0|}\sum_{i\in\widebar{\mathcal{I}}^0}\frac{\bm{a}_i\bm{a}_i^\ccalT}{\|\bm{a}_i\|_2^2}\Big)
	\bm{z}
\end{equation}
whose solution can be well approximated with a few (e.g., $100$) power iterations at a much cheaper computational complexity  $\mathcal{O}(n|\widebar{\mathcal{I}}^0|)$ [than $\mathcal{O}(n^3)$ required for solving \eqref{eq:mineig}]. 
When $\|\bm{x}\|_2\ne 1$ is unknown, $\tilde{\bm{z}}^0$ from~\eqref{eq:maxeig} can be scaled by the norm estimate of $\bm{x}$ to obtain $\bm{z}^0=\sqrt{\sum_{i=1}^m y_i/m}~\tilde{\bm{z}}^0$~\cite{wf,taf}. If $m/n$ is large enough, it has been shown that the orthogonality-promoting initialization can produce an estimate of any given constant relative error~\cite{taf}. 

When $\bm{x}$ is \emph{a priori} known to be $k$-sparse with $k\ll n$, one may expect to recover $\bm{x}$ from a significantly smaller number ($\ll n$) of measurements. The orthogonality-promoting initialization (and spectral based alternatives) requiring $m$ to be on the order of $n$ would fail in the case of PR for sparse signals given a small number of measurements~\cite{tsp2015njs,wf,twf,taf,reshaped1}. By accounting for the sparsity prior information with the $\ell_0$ regularization, the same rationale as the orthogonality-promoting initialization in~\eqref{eq:mineig} would lead to 
\begin{equation}
\label{eq:spca}
		\underset{\|\bm{z}\|_2=1}{\minimize}~~
\bm{z}^\ccalT\bm{Y}
	\bm{z}\quad
	{\rm subject~to}~~\|\bm{z}\|_0=k.
\end{equation}
The problem at hand is \emph{NP-hard} in general due to the combinatorial constraint. Additionally, it can not be readily converted to a (sparse) PCA problem since the number of data samples available is much smaller than the signal dimension $n$, thus hardly validating the non-asymptotic result in \eqref{eq:concen}. Although at much higher computational complexity than power iterations, semidefinite relaxation could be applied~\cite{siam2007agjl}. Instead of coping with~\eqref{eq:spca} directly, we shall take another route and develop our sparse orthogonality-promoting initialization approach to obtain a meaningful sparse initialization from the given limited number of measurements.

\subsubsection{Exact support recovery}

Along the lines of sparse AltMinPhase and sparse PCA~\cite{as2009am}, our approach is to first estimate the support of the underlying signal based on a carefully-designed rule; next, we will rely on power iterations to solve~\eqref{eq:maxeig} restricted on the estimated support, thus ensuring a $k$-sparse estimate $\tilde{\bm{z}}^{0}\in\mathbb{R}^n$; and, subsequently we will scale $\tilde{\bm{z}}^{0}$ by the $\bm{x}$ norm estimate $\sqrt{\sum_{i=1}^m y_i/m}$ to yield a $k$-sparse orthogonality-promoting initialization $\bm{z}^0$.

Starting with the support recovery procedure, assume without loss of generality that $\bm{x}$ is supported on $\mathcal{S}\subseteq [n]:=\{1,\,\ldots,\,n\}$ with $|\mathcal{S}|=k\ll n$. Consider the random variables $Z_{i,j}:=\psi_i^2a_{i,j}^2$, $j=1,\ldots,n$.
Recalling that for standardized Gaussian variables, we have $\mathbb{E}[a_{i,j}^4]=3$, $\mathbb{E}[a_{i,j}^2]=1$, the rotational invariance property of Gaussian distributions confirms for all $1\le j\le n$ that
\begin{align}
\label{eq:exp}
	\mathbb{E}[Z_{i,j}]=\mathbb{E}\big[(\bm{a}_i^\ccalT\bm{x})^2a_{i,j}^2\big]&=\mathbb{E}\big[a_{i,j}^4 x_j^2+
(\bm{a}_{i,/j}^\ccalT\bm{x}_{/j})^2a_{i,j}^2\big]\nonumber\\
&=3x_j^2+\|\bm{x}_{/j}\|_2^2\nonumber\\
&=2x_j^2+\|\bm{x}\|_2^2
\end{align}
where $\bm{x}_{/j}\in\mathbb{R}^{n-1}$ is obtained by deleting the $j$-th entry from $\bm{x}\in\mathbb{R}^n$; and likewise for $\bm{a}_{i,/j}\in\mathbb{R}^{n-1}$. 
If $j\in \mathcal{S}$, then $x_j\ne 0$ yielding   
$\mathbb{E}[Z_{i,j}]=\|\bm{x}\|_2^2+2x_j^2$ in~\eqref{eq:exp}. If on the other hand $j\notin \mathcal{S}$, it holds that $x_j=0$, which leads to $\mathbb{E}[Z_{i,j}]=\|\bm{x}_{/j}\|_2^2=\|\bm{x}\|_2^2$. It is now clear that there is a separation of $2x_j^2$ in the expected values of $Z_{i,j}$ for $j\in \mathcal{S}$ and $j\notin \mathcal{S}$. As long as the gap $2x_j^2$ is sufficiently large, the support set $\mathcal{S}$ can be recovered exactly in this way. Specifically, when all $\mathbb{E}[Z_{i,j}]$ values are available, the set of indices corresponding to the $k$-largest $\mathbb{E}[Z_{i,j}]$ values recover exactly the support of $\bm{x}$. In practice, $\{\mathbb{E}[Z_{i,j}]\}$ are not available.  
One has solely access to a number of their independent realizations. Appealing to the strong law of large numbers, the sample average approaches the ensemble one, namely, $\hat{Z}_{i,j}:=(1/m)\sum_{i=1}^m Z_{i,j}\to \mathbb{E}[Z_{i,j}]$ as $m$ increases. Hence, the support can be estimated as  
\begin{equation}\label{eq:suppest}
	\hat{\mathcal{S}}:=\big\{1\le j\le n\big|\text{indices of top-$k$ instances in $\{\hat{Z}_{i,j}\}_{j=1}^n$
}\big\}
\end{equation}
which will be shown to recover $\mathcal{S}$ exactly with high probability provided that $\mathcal{O}(k^2\log n)$ measurements are taken and the minimum nonzero entry $x_{\min}:=\min_{j\in S}|x_j|$ is on the order of $(1/\sqrt{k})\|\bm{x}\|_2$. The latter is postulated to guarantee such a separation between quantities having their indices belonging or not belonging to the support set.
It is worth stressing that $k^2\log n\ll n$ when $k\ll n$, hence largely reducing the sampling size and also the computational complexity.

%
%

\subsubsection{Orthogonality-promoting intialization}
When the estimated support in~\eqref{eq:suppest} turns out to be exact, i.e., $\hat{\mathcal{S}}=\mathcal{S}$, 
one can rewrite $\psi_i=|\bm{a}_i^\ccalT\bm{x}|=|\bm{a}_{i,\hat{\mathcal{S}}}^\ccalT\bm{x}_{\hat{\mathcal{S}}}|$,  $i=1,\,\ldots,\,m$, where $\bm{a}_{i,\hat{\mathcal{S}}}\in\mathbb{R}^{k}$ includes the $j$-th entry $a_{i,j}$ of $\bm{a}_i$ if and only if $j\in \hat{\mathcal{S}}$; and likewise for $\bm{x}_{\hat{\mathcal{S}}}\in\mathbb{R}^{k}$. Instead of seeking directly an $n$-dimensional initialization as in~\eqref{eq:spca}, one can apply the orthogonality-promoting initialization steps in~\eqref{eq:mineig}-\eqref{eq:maxeig} on the dimensionality reduced data $\{(\bm{a}_{i,\hat{\mathcal{S}}},\psi_i)\}_{i=1}^m$ to produce a $k$-dimensional vector
\begin{equation}
\label{eq:maxeigk}
	\tilde{\bm{z}}_{\hat{\mathcal{S}}}^0:=\arg	\underset{\|\bm{z}_{\hat{\mathcal{S}}}\|_2=1}{\text{max}}~~
\frac{1}{|\widebar{\mathcal{I}}^0|}\bm{z}_{\hat{\mathcal{S}}}^\ccalT\Big(\sum_{i\in\widebar{\mathcal{I}}^0}\frac{\bm{a}_{i,\hat{\mathcal{S}}}\bm{a}_{i,\hat{\mathcal{S}}}^\ccalT}{\|\bm{a}_{i,\hat{\mathcal{S}}}\|_2^2}\Big)
	\bm{z}_{\hat{\mathcal{S}}}
\end{equation}
and subsequently reconstruct a $k$-sparse $n$-dimensional initialization $\tilde{\bm{z}}^0$ by zero-padding $\tilde{\bm{z}}_{\hat{\mathcal{S}}}^0$ at entries with indices not belonging to $\hat{\mathcal{S}}$.
Similarly, in the case of $\|\bm{x}\|_2\ne 1$, $\tilde{\bm{z}}^0$ in~\eqref{eq:maxeigk} is rescaled by the norm estimate of $\bm{x}$ to obtain $\bm{z}^0=\sqrt{\sum_{i=1}^m y_i/m}\,\tilde{\bm{z}}^0$. 
We also note that our proposed algorithm can recover the underlying sparse signal when $\hat{\mathcal{S}} \neq \mathcal{S}$, as long as $\bm{z}^0$ is sufficiently close to $\bm{x}$ regardless of support mismatch, which is described further in Lemma \ref{le:thresh}.

\begin{algorithm}[t]
  \caption{SPARse Truncated Amplitude flow (SPARTA)}
  \label{alg:SPARTA}
  \begin{algorithmic}[1]
\STATE {\bfseries Input:}
Data $\{(\bm{a}_i;\psi_i)\}_{i=1}^m$ and sparsity level $k$; 
 maximum number of iterations $T=1,000$; step size $\mu=1$, truncation thresholds $|\widebar{\mathcal{I}}^0|=\lceil\frac{1}{6}m\rceil$, 
and $\gamma=1$.
\STATE{\bfseries Set} $\hat{\mathcal{S}}$ to include indices 
corresponding to the $k$-largest instances in $\big\{\sum_{i=1}^m\psi_i^2|a_{i,j}|^2/m\big\}_{j=1}^n$.
\STATE {\bfseries Evaluate}\label{step:3}
 $\widebar{\mathcal{I}}^0$ to consist of indices of the top-$|\widebar{\mathcal{I}}^0|$ values in $\{\psi_i/\|\bm{a}_{i,\hat{\mathcal{S}}}\|_2\}_{i=1}^m$ with $\bm{a}_{i,\hat{\mathcal{S}}}\in\mathbb{R}^k$ removing entries of $\bm{a}_i\in\mathbb{R}^n$ not belonging to $\hat{S}$; and compute the principal eigenvector $\tilde{\bm{z}}^0_{\hat{\mathcal{S}}}\in\mathbb{R}^k$ of matrix $$\bm{Y}:=\frac{1}{|\widebar{\mathcal{I}}^0|}
\sum_{i\in\widebar{\mathcal{I}}^0}\frac{\bm{a}_{i,\hat{\mathcal{S}}}\bm{a}_{i,\hat{\mathcal{S}}}^\ccalT}{\|\bm{a}_{i,\hat{\mathcal{S}}}\|_2^2}$$ based on $100$ power iterations.
\STATE {\bfseries Initialize}
\label{step:4} $\bm{z}^0$ as $\sqrt{\sum_{i=1}^m\psi_i^2/m}\,\tilde{\bm{z}}^0$, where $\tilde{\bm{z}}^0\in\mathbb{R}^n$ is obtained by augmenting $\tilde{\bm{z}}^0_{\hat{\mathcal{S}}}$ in Step $3$ with zeros at entries with their indices not in $\hat{\mathcal{S}}$.
  \STATE {\bfseries Loop: For}\label{step:5} 
  {$t=0$ {\bfseries to} $T-1$}\\
\vspace{-.em}
{ \begin{equation*}\vspace{-.em}
	   	\bm{z}^{t+1}=\mathcal{H}_k\!\bigg(\bm{z}^t-\frac{\mu}{m}\sum_{i\in\mathcal{I}^{t+1}}\Big(\bm{a}_i^\ccalT\bm{z}^t-\psi_i\frac{\bm{a}_i^\ccalT\bm{z}^t}{|\bm{a}_i^\ccalT\bm{z}^t|}\Big)\bm{a}_i\bigg)
\end{equation*}
   where $\mathcal{I}^{t+1}=\left\{1\le i\le m\big|{|\bm{a}_i^\ccalT\bm{z}^t|}\ge \psi_i\big/(1+\gamma)\right\}$, and $\mathcal{H}_k(\bm{u}):\mathbb{R}^n\to \mathbb{R}^n$ sets all entries of $\bm{u}$ to zero except for the $k$-ones of largest magnitudes. }
     \STATE {\bfseries Output:}
$\bm{z}^{T}$.
  \end{algorithmic}
\vspace{-.em}
\end{algorithm} 


\vspace{-.em}
\subsection{Thresholded Truncated Gradient Stage}
\label{sub:refine}

Upon obtaining a sparse orthogonality-promoting initialization $\bm{z}^0$, our approach to solving~\eqref{eq:cost} boils down to iteratively refining $\bm{z}^0$ by means of a series of $k$-sparse hard thresholding based truncated gradient iterations, namely, 
\begin{equation}
\label{eq:iteration}
	\bm{z}^{t+1}:=\mathcal{H}_k\!\left(\bm{z}^t-\mu\nabla\ell_{\rm tr}(\bm{z}^t)
	\right),\quad t=0,\,1,\,\ldots
\end{equation}
where $t$ is the iteration index, $\mu>0$ a constant step size, and $\mathcal{H}_k(\bm{u}):\mathbb{R}^n\to \mathbb{R}^n$ denotes a $k$-sparse hard thresholding operation that sets all entries in $\bm{u}$ to zero except for the $k$ entries of largest magnitudes. If there are multiple such sets comprising the $k$-largest entries, a set can be chosen either randomly or according to a predefined ordering of the elements. %
Similar to~\cite{taf}, the truncated (generalized) gradient $\nabla\ell_{\rm tr}(\bm{z}^t)$ is 
\begin{equation}\vspace{-.em}
	\label{eq:tgg}
	\nabla\ell_{\rm tr}(\bm{z}^t):=\frac{1}{m}\sum_{i\in\mathcal{I}^{t+1}}\Big(\bm{a}_i^\ccalT\bm{z}^t-\psi_i\frac{\bm{a}_i^\ccalT\bm{z}^t}{|\bm{a}_i^\ccalT\bm{z}^t|}\Big)\bm{a}_i
\end{equation}
where the index set is defined to be
\begin{equation}
\label{eq:trunc}
\mathcal{I}^{t+1}:=\Big\{1\le i\le m\Big|\frac{|\bm{a}_i^\ccalT\bm{z}^t|}{|\bm{a}_i^\ccalT\bm{x}|}\ge \frac{1}{1+\gamma}\Big\}	
\end{equation}
for some $\gamma>0$ to be determined shortly, where $\{|\bm{a}_i^\ccalT\bm{x}|=\psi_i\}$ are the given modulus data.

 It is clear now that the difficulty of minimizing our nonconvex objective function reduces to that of correctly estimating the signs of $\bm{a}_i^\ccalT\bm{x}$ by ${\bm{a}_i^\ccalT\bm{z}^t}/{|\bm{a}_i^\ccalT\bm{z}^t|}$ at each iteration. 
The truncation rule in \eqref{eq:trunc} was shown capable of eliminating most ``bad'' gradient components involving erroneously estimated signs, i.e., ${\bm{a}_i^\ccalT\bm{z}^t}/{|\bm{a}_i^\ccalT\bm{z}^t|}\ne {\bm{a}_i^\ccalT\bm{x}}/{|\bm{a}_i^\ccalT\bm{x}|}$. This rule improved performance of TAF~\cite{taf} considerably. 
Recall that our objective function in~\eqref{eq:cost} is also non-smooth at points $\bm{z}\in\mathbb{R}^n$ obeying $\bm{a}_i^\ccalT\bm{z}=0$. 
Evidently, the gradient regularization rule in~\eqref{eq:trunc} keeps only the gradients of component functions (i.e., the summands in~\eqref{eq:cost}) that bear large enough $|\bm{a}_i^\ccalT\bm{z}^t|$ values; this rule thus maintains $\bm{a}_i^\ccalT\bm{z}^t$ away from $0$ and protects the cost function in~\eqref{eq:cost} from being non-smooth at points satisfying $\bm{a}_i^\ccalT\bm{z}=0$. As a consequence, the (truncated) generalized gradient employed in~\eqref{eq:tgg} reduces to the (truncated) gradient at such points, which also simplifies theoretical convergence analysis. 

\section{Main Results} 
\label{sec:main}
The proposed sparse phase retrieval solver is summarized in Algorithm~\ref{alg:SPARTA} along with default parameter values. 
Given data samples $\{(\bm{a}_i;\psi_i)\}_{i=1}^m$ generated from i.i.d. $\{\bm{a}_i\}_{i=1}^m\sim\mathcal{N}(\bm{0},\bm{I}_n)$ sensing vectors, the following result establishes the statistical convergence rate for the proposed SPARTA algorithm in the case of $\gamma=+\infty$.

\begin{theorem}[\bf Exact recovery]
	\label{thm:exact}
	Fix $\bm{x}\in\mathbb{R}^n$ to be any $k$-sparse ($k\ll n$) vector of the minimum nonzero entry on the order of $(1/\sqrt{k})\|\bm{x}\|_2$, namely, $x_{\min}^2=(C_1/k)\|\bm{x}\|_2^2$ for some number $C_1>0$. Consider the $m$ noiseless measurements $\psi_i=|\bm{a}_i^\ccalT\bm{x}|$ from i.i.d. $\bm{a}_i\sim \mathcal{N}(\bm{0},\bm{I}_n)$, $1\le i\le m$. If $m\ge C_0k^2\log (mn)$, Step 3 of SPARTA (tabulated in Algorithm \ref{alg:SPARTA}) recovers the support of $\bm{x}$ exactly with probability at least $1-6/m$. 	 Furthermore, there exist numerical constants $\underline{\mu},\;\widebar{\mu}>0$ such that with a fixed step size $\mu\in [\underline{\mu},\,\widebar{\mu}]$, and a truncation threshold $\gamma=+ \infty$, successive estimates of SPARTA obey 
	\begin{equation} 
\label{eq:thm}
		{\rm dist}( \bm{z}^t,\bm{x})\le \frac{1}{10}\left(1-\nu\right)^t\left\|\bm{x}\right\|_2,\quad  t=0,\,1,\,\ldots
	\end{equation} 
	which holds with probability exceeding 
 $1-c_1m{\rm e}^{-c_0k}-7/m$ provided that $m\ge C_2|\widebar{\mathcal{I}}^0|\ge C_0 k^2\log (mn)$.
	Here, $c_0,\,c_1,\,C_0,\,C_2$, and $0<\nu<1$ are some numerical constants. 
\end{theorem}

Proof of Theorem~\ref{thm:exact} is deferred to Section~\ref{sec:proof} with supporting lemmas presented in the Appendix. 
We typically take parameters $|\widebar{\mathcal{I}}^0|=\lceil\frac{1}{6}m\rceil$, and $\mu=1$, which will also be validated by our analytical results on the feasible region of the step size. 
The constant $C_0$ depends on $C_1$, $\nu$ on $\mu$ and $C_1$, and $\underline{\mu}$ and $\widebar{\mu}$ rely on both $C_1$ and $C_0$.
In the case of PR of unstructured signals, existing algorithms such as TAF ensures exact recovery when the number of measurements $m$ is about the number of unknowns $n$, i.e., $m\gtrsim  n$. Hence, it would be more meaningful to study the sample complexity bound for PR of sparse signals when $m  \lesssim n$. To this end, the sample complexity bound $m\ge C_0k^2\log(mn)$ in Theorem~\ref{thm:exact} can often be rewritten as $m\ge C_0'k^2\log n$ for some constant $C_0'>C_0$ and large enough $n$. 
Regarding Theorem~\ref{thm:exact}, three observations are in order.

 \begin{remark}
 SPARTA recovers exactly any $k$-sparse signal $\bm{x}$ of minimum nonzero entries on the order of $(1/\sqrt{k})\|\bm{x}\|_2$ when there are about $  k^2\log n$ magnitude-only measurements, which coincides with the number of measurements required by the state-of-the-art algorithms such as CPRL~\cite{cprl}, sparse AltMinPhase~\cite{tsp2015njs}, and TWF~\cite{as2016clm}. 

 \end{remark}

\begin{remark}
SPARTA converges at a linear rate to the globally optimal solution $\bm{x}$ with convergence rate independent of the signal dimension $n$. In other words, for any given solution accuracy $\epsilon>0$, after running at most $T=\log(1/\epsilon)$  SPARTA iterations~\eqref{eq:iteration}, the returned estimate $\bm{z}^T$ is at most $\epsilon\|\bm{x}\|_2$ away from the global solution $\bm{x}$.	
\end{remark}

\begin{remark}
SPARTA enjoys a low computational complexity of $\mathcal{O}(k^2n\log n)$, and incurs a total runtime of $\mathcal{O}(k^2n\log n\log(1/\epsilon))$ to produce an $\epsilon$-accurate solution. The runtime is proportional to the time $\mathcal{O}(k^2n\log n)$ taken to read the data $\{(\bm{a}_i,\psi_i)\}_{i=1}^m$. 
To see this, recall that the support recovery incurs computational complexity $\mathcal{O}(k^2n\log n+n\log n)$, power iterations incur complexity $\mathcal{O}(k^2n\log n )$, and thresholded truncated gradient iterations have complexity $\mathcal{O}(k^2n\log n)$; hence, leading to a total complexity on the order of $k^2n\log n$. Given the linear convergence rate, SPARTA takes a total runtime of $\mathcal{O}(k^2n\log n\log(1/\epsilon))$ to achieve any fixed solution accuracy $\epsilon>0$.
\end{remark}

Besides exact recovery guarantees in the case of noiseless measurements, it is worth mentioning that SPARTA exhibits robustness to additive noise, especially when the noise has bounded values. Numerical results using SPARTA for noisy sparse PR will be presented in the ensuing section.

\section{Numerical Experiments}
\label{sec:test}
Simulated tests evaluating performance of SPARTA relative to truncated amplitude flow (TAF)~\cite{taf} (which does not exploit the sparsity) and thresholded Wirtinger flow (TWF) 
 \cite{as2016clm} 
 are presented in this section. 
For fair comparisons, the algorithmic parameters involved in all schemes were set to their suggested values. 
The initialization in each scheme was obtained based upon $100$ power iterations, and was subsequently refined by $T=1,000$ gradient iterations. 
In all reported experiments, the true $k$-sparse signal vector $\bm{x}\in\mathbb{R}^n$ or $\mathbb{C}^n$ was generated first using $\bm{x}\sim\mathcal{N}(\bm{0},\bm{I}_n)$ or $\mathcal{CN}(\bm{0},\bm{I}_n)$, followed by setting $(n-k)$ of its $n$ entries to zero uniformly at random. 
For reproducibility, the Matlab implementation of SPARTA is publicly available at \url{https://gangwg.github.io/SPARTA/}.

The first experiment evaluates the exact recovery performance of various approaches in terms of the empirical success rate over $100$ independent Monte Carlo trials, where the true signals are real-valued. A success is declared for a trial provided that the returned estimate incurs a relative mean-square error defined as 
$${\rm Relative~MSE}:=\frac{{\rm dist}(\bm{z}^T,\bm{x})}{\|\bm{x}\|_2}$$
 less than $10^{-5}$. We fixed the signal dimension to $n=1,000$, and the sparsity level at $k=10$, while the number of measurements $m/n$ increases from $0.1$ to $3$ by $0.1$. Curves in Fig.~\ref{fig:rate} clearly demonstrate markedly improved performance of SPARTA over state-of-the-art alternatives. Even when the exact number of nonzero elements in $\bm{x}$, namely, $k$ is unknown, setting $k$ in Algorithm~\ref{alg:SPARTA} as an upper limit on the theoretically affordable sparsity level (e.g., $\lceil\sqrt{n}\;\rceil$ when $m$ is about $n$ according to Theorem~\ref{thm:exact}) works well too (see the magenta curve, denoted SPARTA0). Comparison between TAF and SPARTA shows the advantage of exploiting sparsity in sparse PR settings. 
\begin{figure}[h!]
	\centering
	\includegraphics[scale=0.6]{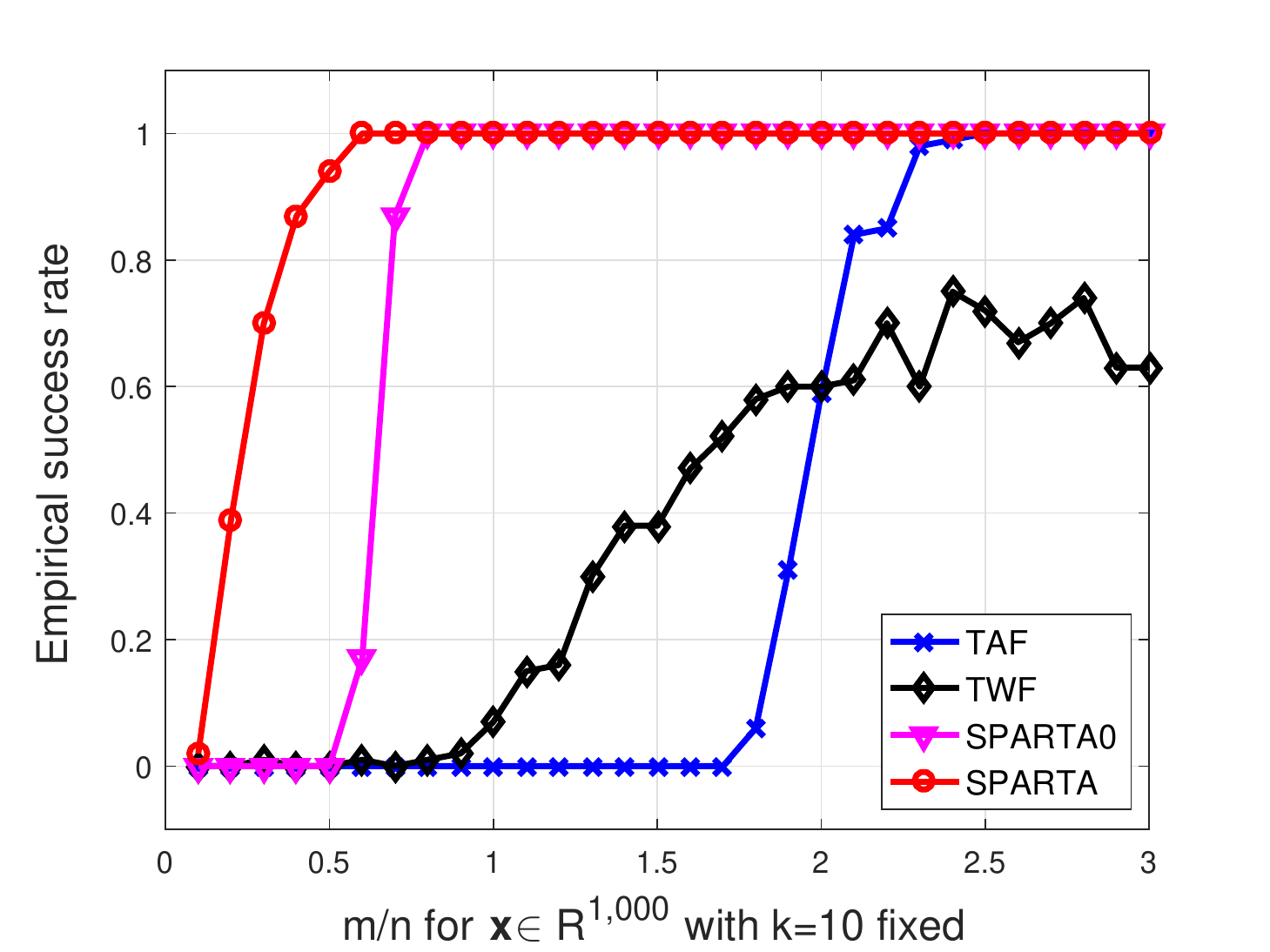}
\caption{
Empirical success rate versus $m/n$ for $\bm{x}\in\mathbb{R}^{n}$ with $n=1,000$
and $k=10$ nonzero entries
using: i) TAF without exploiting sparsity~\cite{taf}; ii) TWF~\cite{as2016clm}; iii) SPARTA0 with the exact number of nonzeros unknown, and $k$ taken as an upper limit $\lceil\sqrt{n}\rceil=32$; and iv) SPARTA with $k=10$. 
}
\label{fig:rate}
\end{figure}

The second experiment examines how SPARTA recovers real-valued signals of various sparsity levels given a fixed number of measurements.   
Figure~\ref{fig:sparsity} depicts the empirical success rate versus the sparsity level $k$, where $k$ equals the exact number of nonzero entries in $\bm{x}$.
The results suggest that with a total of $m=n$ phaseless quadratic equations, 
TAF representing the state-of-the-art for PR of unstructured signals fails,
 as shown by the blue curve. Although TWF works in some cases, SPARTA significantly outperforms TWF, and it ensures exact recovery of sparse signals with up to about $25< \sqrt{n}\approx 32$ nonzero entries (due to existence of polylog factors in the sample complexity), hence justifying our analytical results. 

\vspace{-.em}
\begin{figure}[h!]
	\centering
	\includegraphics[scale=0.6]{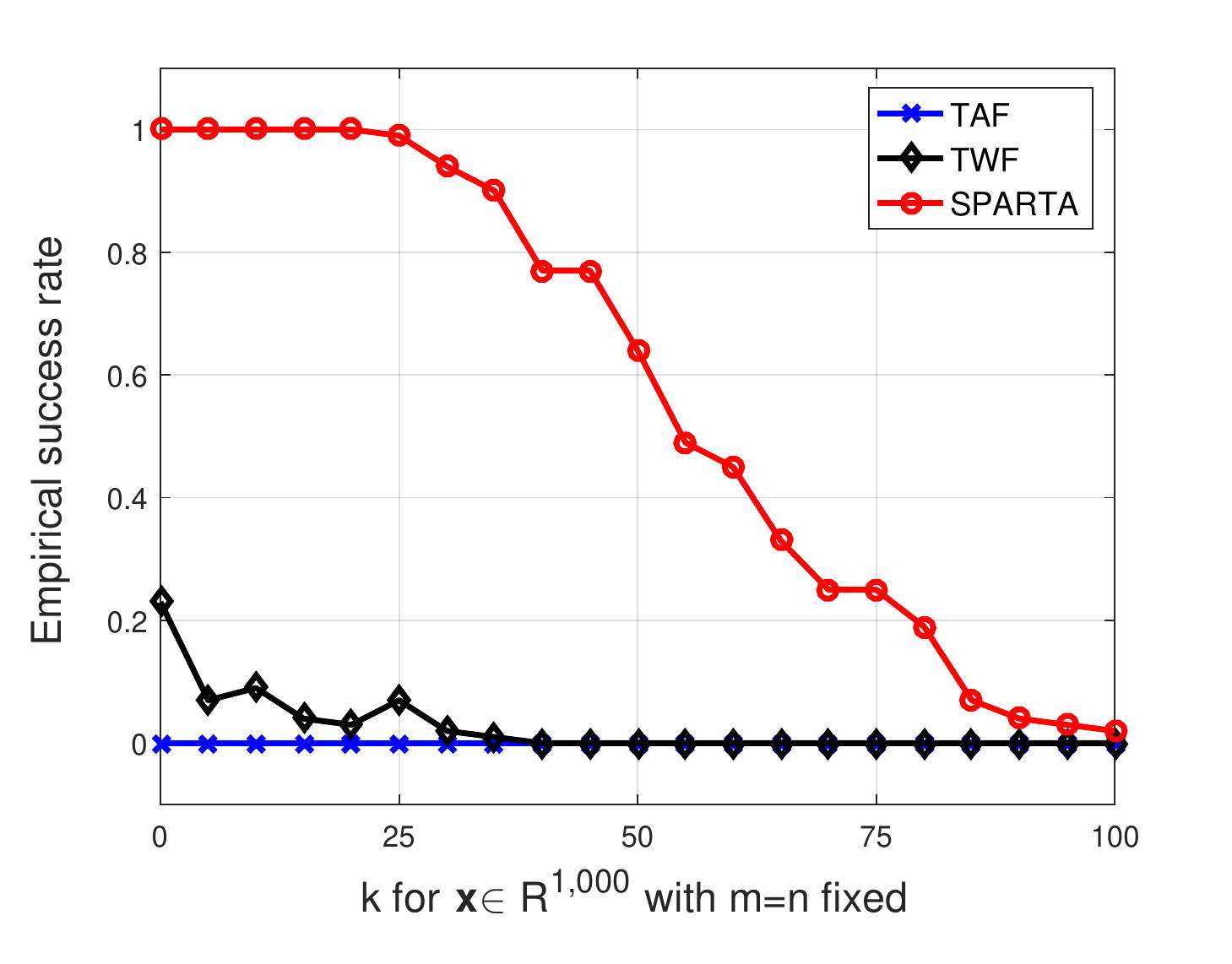}
\caption{
Empirical success rate versus sparsity level $k$ for $\bm{x}\in\mathbb{R}^{n}$ with $m=n=1,000$ fixed
using: i) TAF; ii) TWF, and iii) SPARTA. 
}
\label{fig:sparsity}
\vspace{-.em}
\end{figure}

 \begin{figure}[h!]
 	\centering
 	\includegraphics[scale=0.6]{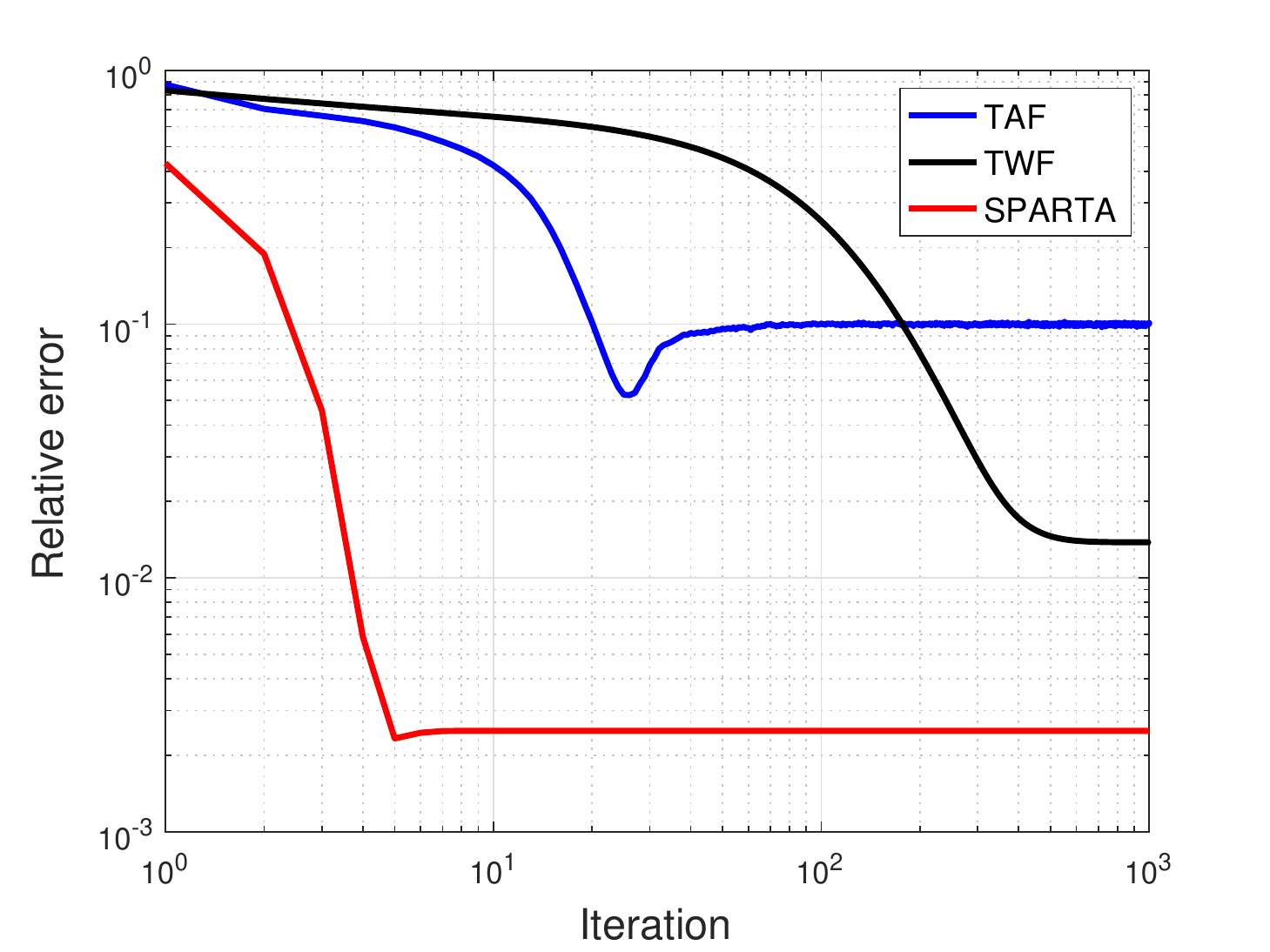}
 	\caption{
 		Convergence behavior in the case of noisy data with $n=1,000$, $m=3,000$, and $k=10$ 
 		using: i) TAF; ii) TWF; and iii) SPARTA. 
 	}
 	\label{fig:noisy}
 \end{figure}
 
The next experiment validates the robustness of SPARTA against additive noise present in the data. Postulating the noisy Gaussian data model $\psi_i=|\bm{a}_i^\ccalT\bm{x}|+\eta_i$~\cite{tsp2015njs}, 
we generated i.i.d. Gaussian noise according to $\eta_i\sim\mathcal{N}(0,\,0.1^2)$, $i=1,\,\ldots,\,m$. From Fig.~\ref{fig:rate}, it is clear that to achieve exact recovery, SPARTA requires about $m=6k^2=600$ measurements, TAF about $3n=3,000$ measurements, and TWF much more than $3,000$. In this case, parameters were taken as $n=1,000$, $m=3,000$, and $k=10$, with the number of measurements large enough to guarantee that TWF and TAF also work. It is worth mentioning that SPARTA can work with a far smaller number of measurements than $m=3,000$. 
As seen from the plots, SPARTA performs only a few gradient iterations to achieve the most accurate solution among the three approaches, while its competing TAF and TWF require nearly an order more number of iterations to converge to less accurate estimates.

To demonstrate the stability of SPARTA in the presence of additive noise, the relative MSE 
is plotted as a function of the signal-to-noise (SNR) values in dB. Our experiments are based on the additive Gaussian noise model $\psi_i=|\bm{a}_i^\ccalT\bm{x}|+\eta_i$ with a $10$-sparse signal $\bm{x}\in\mathbb{R}^{1,000}$ and the noise $\bm{\eta}:=[\eta_1~\cdots~\eta_m]^\ccalT\sim \mathcal{N}(\bm{0},\sigma^2\bm{I}_{m})$, where the variance $\sigma^2$ is chosen such that certain ${\rm SNR}:=10\log_{10} {\sum_{i=1}^m\nicefrac{|\langle\bm{a}_i,\bm{x}\rangle|^2}{ \sigma^2}} 
  $ values are achieved. 
  The ratio $m/n$ takes values $\{1,\,2,\,3\}$, and the SNR in dB
is varied from $5$ dB to $55$ dB. Averaging over $100$ Monte Carlo realizations, Fig.~\ref{fig:db} demonstrates that the relative MSE for all $m/n$ values scales inversely proportional to SNR, hence corroborating the stability of SPARTA in the presence of additive noise.

\begin{figure}[ht]
	\centering
	\includegraphics[scale=0.6]{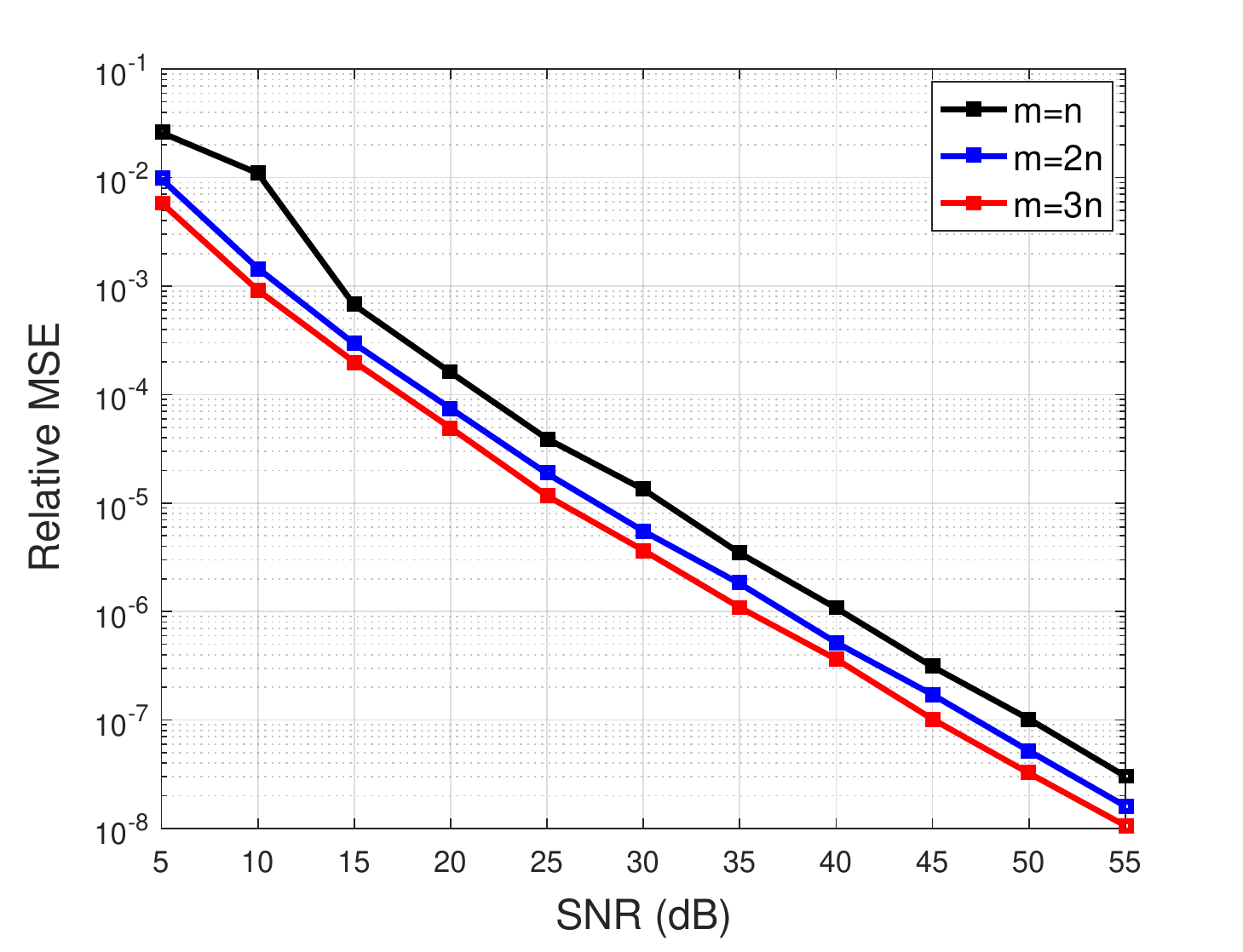}
	\caption{Relative MSE versus SNR for SPARTA with the AWGN model.}
	\label{fig:db}
\end{figure}

The last experiment tested the efficacy of SPARTA in the complex-valued setting, where the underlying $10$-sparse signal $\bm{x}\in\mathbb{C}^{20,000}$ was generated using $\bm{x}\sim\mathcal{CN}(\bm{0},\bm{I}_{20,000}):=\mathcal{N}(\bm{0},\bm{I}_{20,000}/2)+ j \mathcal{N}(\bm{0},\bm{I}_{20,000}/2)$, and the design vectors $\bm{a}_i\sim \mathcal{CN}(\bm{0},\bm{I}_{20,000}) $ for $1\le i \le 1,000$. The relative MSE versus iteration count was plotted in Fig. \ref{fig:rmse}, which validates the scalability and effectiveness of SPARTA in recovering complex signals. 
In terms of runtime, SPARTA recovers exactly a $20,000$-dimensional complex-valued signal from  $1,000$ magnitude-only measurements in a few seconds.

\begin{figure}[ht]
	\centering
	\includegraphics[scale=0.6]{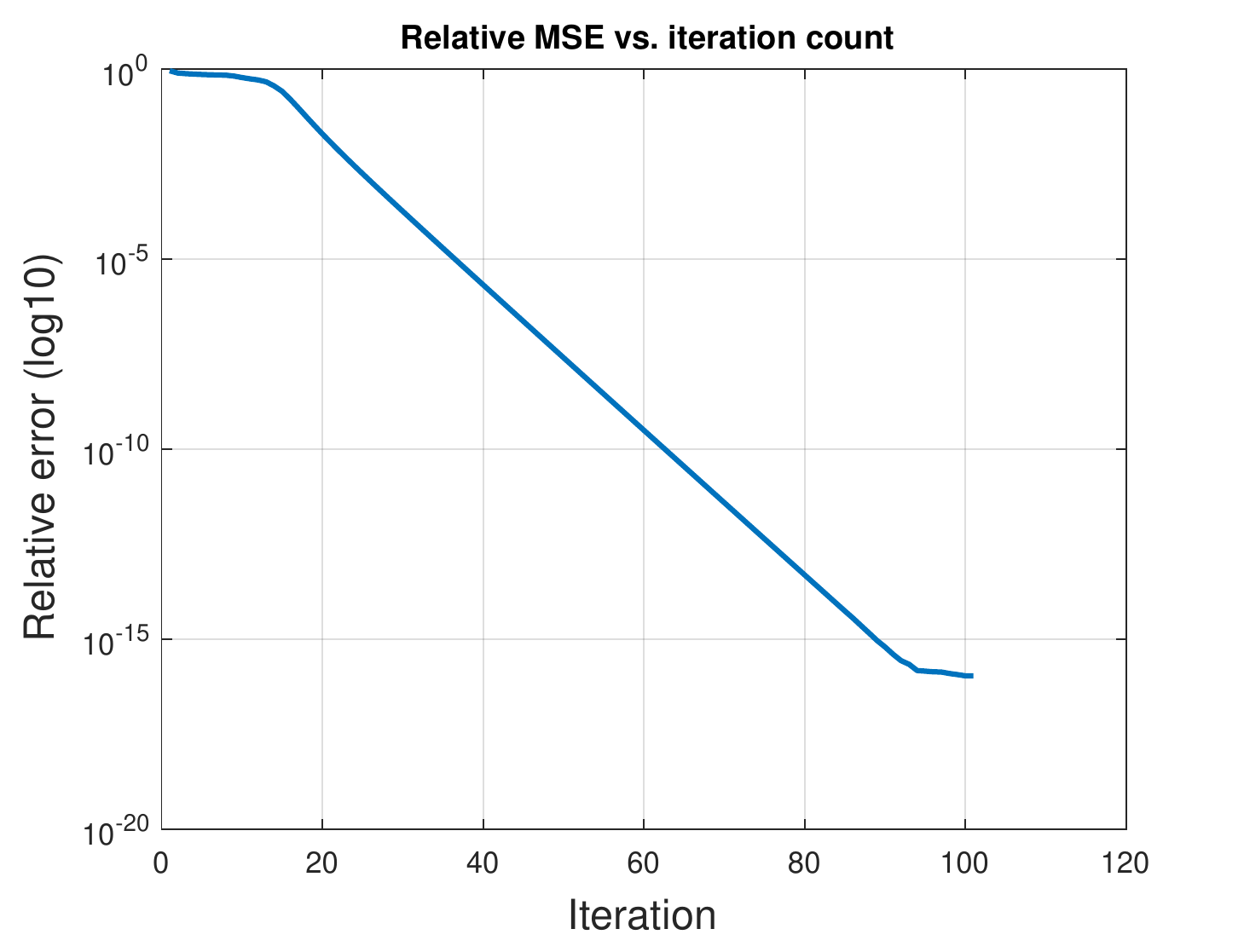}
	\caption{Relative MSE versus iteration count for SPARTA in the complex-valued setting.}
	\label{fig:rmse}
\end{figure}

Regarding computation times, SPARTA converges much faster (both in time and in the number of iterations required to achieve certain solution accuracy) than TWF and TAF in all reported experiments. All numerical experiments were implemented with MATLAB R$2016$a on an Intel CPU @ $3.4$ GHz ($32$ GB RAM) computer.

\section{Proof of Theorem~\ref{thm:exact}} 
\label{sec:proof}

The proof of Theorem~\ref{thm:exact} will be provided in this section. To that end, we will first evaluate the performance of our sparse orthogonality-promoting initialization. The following result demonstrates that if the number of measurements is sufficiently large (on the order of $k^2$ within polylog factors), 
Step \ref{step:3} of the SPARTA algorithm \ref{alg:SPARTA} reconstructs the support of $\bm{x}$ exactly with high probability.   
\begin{lemma}
\label{le: inverse} 
Consider any $k$-sparse signal $\bm{x}\in\mathbb{R}^n$ with support $\mathcal{S}$ and minimum nonzero entries $x_{\min}:=\min_{j\in \mathcal{S}}|x_j|$ on the order of $(1/\sqrt{k})\|\bm{x}\|_2$. If the sensing vectors $\{\bm{a}_i\}_{i=1}^m$ are i.i.d standard Gaussian, i.e., $\bm{a}_i\sim\mathcal{N}(\bm{0},\bm{I}_n)$, Step \ref{step:3} in Algorithm \ref{alg:SPARTA} recovers $\mathcal{S}$ exactly with probability at least $1-6/m$ provided $m\ge C_0k^2\log(mn)$ for some absolute constant $C_0>0$.
\end{lemma}

\begin{proof}[Proof of Lemma~\ref{le: inverse}]
 As elaborated in Section~\ref{sub:init}, there is a clear separation in the expected values $\mathbb{E}[Z_{i,j}]=\mathbb{E}[\psi_i^2a_{i,j}^2]$ for $j\in \mathcal{S}$ and $j\notin \mathcal{S}$; that is,  
\begin{align}
\mathbb{E}[Z_{i,j}]&=\mathbb{E}\big[(\bm{a}_i^\ccalT\bm{x})^2a_{i,j}^2\big]=\mathbb{E}\big[a_{i,j}^4 x_j^2+
(\bm{a}_{i,/j}^\ccalT\bm{x}_{/j})^2a_{i,j}^2\big]\nonumber\\
&=\left\{\begin{array}{ll}
\|\bm{x}\|_2^2,&\; j\notin \mathcal{S},\\
\|\bm{x}\|_2^2+2x_j^2,&\; j\in \mathcal{S}.
\end{array}\right.
\end{align}
Consider the case of $j\in \mathcal{S}$ first. Based on $\mathbb{E}[a_{i,j}^{2p}]=(2p-1)!!$ with $p$ being a positive integer and the symbol $!!$ denoting the double factorial, $Z_{i,j}$ has second-order moment
\begin{align}
	\label{eq:ins}
	\mathbb{E}[Z_{i,j}^2]&=\mathbb{E}\big[(\bm{a}_i^\ccalT\bm{x})^4a_{i,j}^4\big]\nonumber\\
	&=\mathbb{E}\big[a_{i,j}^8x_j^4+a_{i,j}^4a_{i,\ell\ne j}^4\|\bm{x}_{/j}\|_2^4+6a_{i,j}^6x_j^2a_{i,\ell\ne j}^2\|\bm{x}_{/j}\|_2^2
	\big]\nonumber\\
&=105x_j^4+9\|\bm{x}_{/j}\|_2^4+90x_j^2\|\bm{x}_{/j}\|_2^2\nonumber\\
&= 9\|\bm{x}\|_2^4+24x_j^4+72x_j^2\|\bm{x}\|_2^2
\end{align}
where $\ell\in \{1,\,2,\,\ldots,\,n\}$ is some index from different than $j$. 
Letting $\tilde{Z}_{j}:=\|\bm{x}\|_2^2+2x_j^2-Z_{i,j}$ for all $ j\in \mathcal{S}$, it holds that $\tilde{Z}_j\le \|\bm{x}\|_2^2+2x_j^2\le 3\|\bm{x}\|_2^2$. Furthermore, one has $\mathbb{E}[\tilde{Z}_j]=0$, and 
\begin{align}
\mathbb{E}[\tilde{Z}_j^2]&=\!\|\bm{x}\|_2^4\!+\!4x_j^4\!+\!4x_j^2\|\bm{x}\|_2^2\!+\!\mathbb{E}[Z_{i,j}^2]\!-\!\big(2\|\bm{x}\|_2^2\!+\!4x_j^2\big)\mathbb{E}[Z_{i,j}]	\nonumber\\
&=8\|\bm{x}\|_2^4+68x_j^2\|\bm{x}\|_2^2+20x_j^4\nonumber\\
&\le 96\|\bm{x}\|_2^4.\nonumber
\end{align}

Appealing to Lemma~\ref{le:bounded}, 
one establishes for all $ j\in \mathcal{S}$ that 
\begin{equation}
	{\rm Pr}\Big(\frac{1}{m}	\sum_{i=1}^m\psi_i^2a_{i,j}^2-(\|\bm{x}\|_2^2+2x_j^2)\le -\epsilon\Big)\le \exp\Big(-\frac{m\epsilon^2}{192\|\bm{x}\|_2^4}\Big).\nonumber
\end{equation}
Taking $\epsilon=x_{\min}^2:=\min_{j\in \mathcal{S}} x_j^2\le x_j^2$ leads to
\begin{equation}
	{\rm Pr}\Big(\frac{1}{m}	\sum_{i=1}^m\psi_i^2a_{i,j}^2\le \|\bm{x}\|_2^2+x_{\min}^2\Big)\le \exp\Big(-\frac{mx_{\min}^4}{192\|\bm{x}\|_2^4}\Big).\nonumber
\end{equation}
Recalling our assumption that $x_{\min}^2$ is on the order of $(1/k)\|\bm{x}\|_2^2$, i.e., $x_{\min}^2=(C_1/k)\|\bm{x}\|_2^2$ for certain constant $C_1>0$, 
the following holds with probability at least $1-1/m$ for all $j\in\mathcal{S}$
\begin{align}
\label{eq:1stfinal}
\min_{j\in \mathcal{S}}\frac{1}{m}	\sum_{i=1}^m\psi_i^2a_{i,j}^2\ge \|\bm{x}\|_2^2+x_{\min}^2=\Big(1+\frac{C_1}{k}\Big)\|\bm{x}\|_2^2
\end{align}
provided that $m\ge C_0k^2\log(mn)$ for some absolute constant $C_0>0$.

Now let us turn to the case of $j\notin \mathcal{S}$, in which $\sum_{i=1}^mZ_{i,j}=\sum_{i=1}^m\psi_i^2a_{i,j}^2$ is a weighted sum of $\chi_1^2$ random variables. According to Lemma \ref{le:2000}, it holds that
\begin{align}
	\label{eq:20001}
	{\rm Pr}\Big(\sum_{i=1}^m\psi_i^2(a_{i,j}^2\!-\!1)\!>\!2\sqrt{\epsilon}\big(\sum_{i=1}^m\psi_i^4\big)^{\frac{1}{2}}\!+\!2\epsilon\max_i \psi_i^2\Big)\!\le \!\exp(-\epsilon).
\end{align}
In addition, for any constants $\epsilon',\,\epsilon''>0$, Chebyshev's inequality together with the union bound confirms that
\begin{subequations}
\label{eq:coefficient}
\begin{align}
&	{\rm Pr}\Big(\sum_{i=1}^m\psi_i^4>\big(3m+\sqrt{96m}\epsilon'\big)\|\bm{x}\|_2^4\Big)\le 1/(\epsilon')^2\label{eq:2norm}\\
&	{\rm Pr}\Big(\max_{1\le i\le m}\psi_i^2>\epsilon''\|\bm{x}\|_2^2\Big)\le 2m\exp(-\epsilon''/2).\label{eq:infinity}
\end{align}	
\end{subequations}
Take
 $\epsilon:=\log(mn)$ in~\eqref{eq:20001}, $\epsilon':=\sqrt{m}$ and $\epsilon'':=4\log(mn)$ in~\eqref{eq:coefficient}. Then, with probability at least $1-4/m$, the next holds for all $j\notin \mathcal{S}$ and $m>C'$
\begin{align}\label{eq:1stpart}
\frac{1}{m}\sum_{i=1}^m\psi_i^2(a_{i,j}^2-1)&\le \frac{2}{m}\sqrt{\log(mn)}\sqrt{3m+\sqrt{96m}\sqrt{m}}\|\bm{x}\|_2^2+\frac{8}{m}\big(\log(mn)\big)^2\|\bm{x}\|_2^2\nonumber\\
&\le 8\sqrt{\frac{\log(mn)}{m}}\|\bm{x}\|_2^2
\end{align}
for some absolute constant $C'>0$ depending on $n$. 

On the other hand, the rotational invariance property of Gaussian distributions asserts that $\psi_i^2=|\bm{a}_i^\ccalT\bm{x}|^2=|\bm{a}_{i,\mathcal{S}}^\ccalT\bm{x}_{\mathcal{S}}|^2\overset{d}{=}a_{i,j}^2\|\bm{x}\|_2^2$~\cite{phaselift}, in which the symbol $\overset{d}{=}$ means that terms involved on both sides of the equality enjoy the same distribution. Since the $\chi^2$ variables $a_{i,j}^2$ are sub-exponential, 
an application of Bernstein's inequality produces the tail bound
\begin{equation}
{\rm Pr}\Big(\frac{1}{m}\sum_{i=1}^ma_{i,j}^2-1\ge \epsilon\Big)\le \exp(-m\epsilon^2/8)
\end{equation}   
for any $\epsilon\in(0,1)$, which can also be easily verified with a direct tail probability calculation from the tail probability of standard Gaussian distribution.
Choosing $\epsilon:=\sqrt{16 \log (m)/m}$ with $m>C'$ gives rise to
\begin{equation}\label{eq:2ndpart}
\frac{1}{m}\sum_{i=1}^m\psi_{i,j}^2\le\bigg(1+4\sqrt{\frac{\log m}{m}}\bigg)
 \|\bm{x}\|_2^2
\end{equation}   
which holds true with probability at least $1-{1}/{m}$ for all $ j\in[ m]$. 
Putting results in \eqref{eq:1stpart} and \eqref{eq:2ndpart} together leads to 
\begin{align}\label{eq:2ndfinal}
\max_{j\notin \mathcal{S}\subseteq [m]}\frac{1}{m}\sum_{i=1}^m\psi_i^2a_{i,j}^2\le \bigg(1+12\sqrt{\frac{\log (mn)}{m}}\bigg)
 \|\bm{x}\|_2^2
\end{align}
which holds with probability exceeding $1-5/m$ for large enough $m$. 

The last inequality taken collectively with \eqref{eq:1stfinal} suggests that there exists an event $E_0$ on which with probability at least $1-6/m$, the following holds
\begin{align}
\label{eq:final}
\min_{j\in \mathcal{S}}\frac{1}{m}\sum_{i=1}^m\psi_i^2 a_{i,j}^2&\ge \Big(1+\frac{C_1}{k}\Big)\|\bm{x}\|_2^2\nonumber\\
&>
\bigg(1+12\sqrt{\frac{\log(mn)}{m}}\bigg)\|\bm{x}\|_2^2\nonumber\\
&\ge \max_{j\notin \mathcal{S}}\frac{1}{m}\sum_{i=1}^m\psi_i^2 a_{i,j}^2
\end{align}
provided that $m\ge C_0k^2\log(mn)$ such that $C_0\ge 144/C_1^2$ with $x_{\min}^2=(C_1/k)\|\bm{x}\|_2^2$. 
\end{proof}

Upon obtaining the support of the underlying sparse signal, SPARTA subsequently employs the orthogonality-promoting initialization on the reduced-dimension data $\{(\psi_i,\bm{a}_{i,\hat{S}})\}$. Based on results in~\cite[Proposition 1]{taf}, the estimate $\bm{z}^0_{\hat{S}}:=\sqrt{\sum_{i=1}^m\psi_i^2/m}\tilde{\bm{z}}^0_{\hat{S}}$ obtained from Step $3$ in Algorithm \ref{alg:SPARTA} satisfies ${\rm dist}(\bm{z}^0_{\hat{S}},\bm{x}_{\hat{S}})\le (1/10)\|\bm{x}_{\hat{S}}\|_2$ with high probability provided that $m/k$ is sufficiently large and $k$ large enough as well. 
Putting together this result, Lemma~\ref{le: inverse}, and Step $4$ in Algorithm \ref{alg:SPARTA} leads to the following lemma, which formally summarizes the theoretical performance of our proposed sparse orthogonality-promoting initialization.

\begin{lemma}
	\label{le:dist}
	Let ${\bm{z}}_0=\sqrt{\sum_{i=1}^m\psi_i^2/m}\,\tilde{\bm{z}}^0$ be given by Step $4$, and $\tilde{\bm{z}}^0$ obtained through the sparse orthogonality-promoting initialization Step $3$ in Algorithm~\ref{alg:SPARTA}. With probability at least $1-(m+6)\exp(-k/2)-7/m$, the following holds
\begin{equation}
{\rm dist}(\bm{z}_0,\bm{x})\le (1/10)\|\bm{x}\|_2\label{eq:distance}
\end{equation}
	provided that $m\ge C_0'k$ for some absolute constant $C_0'>0$. 
\end{lemma}

The proof can be directly adapted from \cite[Proposition 1]{taf}, and hence it is omitted.

\begin{lemma}
\label{le:thresh}
 Take a constant learning parameter  $\mu\in(\underline{\mu},\,\widebar{\mu})$. 
  There exists an event of probability at least $1-c_1m^{-c_0k}$, such that on this event, starting from an initial estimate $\bm{z}^0$ satisfying ${\rm dist}(\bm{z}^0,\bm{x})\le (1/10)\|\bm{x}\|_2$, successive estimates by Step \ref{step:5} with $\gamma=+\infty$ in Algorithm~\ref{alg:SPARTA} obey 
 \begin{equation}
 {\rm dist}(\bm{z}^t,\bm{x})\le (1/10)(1-\nu)^t\|\bm{x}\|_2
,\quad t=0,\,1,\,\ldots
 \end{equation} 
if $m\ge C_0'' (3k)\log(n/(3k))$. Here, $\underline{\mu},\,\widebar{\mu}_0,\,c_0,\,c_1,\,C_0''>0$ are certain universal constants. 
\end{lemma}

It is worth noting that Step \ref{step:5} of Algorithm \ref{alg:SPARTA} guarantees linear convergence to the globally optimal solution $\bm{x}$ as long as the initial guess $\bm{z}^0$ lands within a small neighborhood of $\bm{x}$, \emph{regardless of whether} $\bm{z}^0$ estimates exactly the support of $\bm{x}$ or not. 

\begin{proof}[Proof of Lemma~\ref{le:thresh}]

To start, let us establish a bit of notation, which will be used only in this section. Define for all $t\ge 0$
$$\bm{d}^{t+1}:=\bm{z}^t-\frac{\mu}{m}\sum_{i=1}^m\Big(\bm{a}_i^\ccalT\bm{z}^t-\psi_i\frac{\bm{a}_i^\ccalT\bm{z}^t}{|\bm{a}_i^\ccalT\bm{z}^t|}\Big)\bm{a}_i$$ which represents the estimate prior to the hard thresholding operation in~\eqref{eq:iteration}. 
With $\mathcal{S}$ and $\hat{\mathcal{S}}^{t}$ denoting the support set of $\bm{x}$ and $\bm{z}^{t}$, respectively,
the reconstruction error $\bm{x}-\bm{z}^{t+1}$ is therefore supported on the set $\Theta^{t+1}:=\mathcal{S}\cup \hat{\mathcal{S}}^{t+1}$;
and likewise, $\bm{x}-\bm{z}^t$ is supported on $\Theta^{t}:=\mathcal{S}\cup \hat{\mathcal{S}}^{t}$. In addition, define the difference between sets $\Theta^{t}$ and $\Theta^{t+1}$ as $\Theta^{t}\setminus \Theta^{t+1}$, which consists of all elements of $\Theta^{t}$ that are not elements of $\Theta^{t+1}$.
It is then clear that $|\mathcal{S}|=|\hat{\mathcal{S}}^{t}|=k$, $|\Theta^t|\le 2k$, and $|\Theta^{t}\setminus \Theta^{t+1}|\le 2k$ as well as $|\Theta^{t}\cup \Theta^{t+1}|\le 3k$ for all $t\ge 0$. When using these sets as subscript, for instance, $\bm{d}_{\Theta^t}$, we mean vectors formed by deleting all but those elements from the vector other than those in the set.

The proof of Lemma~\ref{le:thresh} will be mainly based on results in \cite{taf}, and \cite{acha2009bd}, \cite{acha2009nt}. 
The former helps establishing the so-termed local regularity condition that will be key to proving linear convergence of iterative optimization algorithms to the globally optimal solutions of nonconvex optimization problems~\cite{wf}, while the latter two offer a standard approach to dealing with the nonlinear hard thresholding operator involved in our proposed SPARTA algorithm.
Specifically, 
based on the triangle inequality of the vector $2$-norm, one arrives at
\begin{align}
\big\|\bm{x}_{\Theta^{t+1}}-\bm{z}^{t+1}_{\Theta^{t+1}}\big\|_2&=\big\|\bm{x}_{\Theta^{t+1}}-\bm{d}^{t+1}_{\Theta^{t+1}}+\bm{d}^{t+1}_{\Theta^{t+1}}-\bm{z}^{t+1}_{\Theta^{t+1}}\big\|_2\nonumber\\
&\le \big\|\bm{x}_{\Theta^{t+1}}-\bm{d}^{t+1}_{\Theta^{t+1}}\big\|_2+\big\|\bm{z}^{t+1}_{\Theta^{t+1}}-\bm{d}^{t+1}_{\Theta^{t+1}}\big\|_2\label{eq:triangle}
\end{align}
where in the last inequality the first term denotes the distance of $\bm{x}_{\Theta^{t+1}}$ to the estimate $\bm{d}^{t+1}_{\Theta^{t+1}}$ before hard thresholding, and the second denotes the distance between $\bm{d}^{t+1}_{\Theta^{t+1}}$ and its best $k$-approximation $\bm{z}^{t+1}_{\Theta^{t+1}}$ because $\bm{z}_{\Theta^{t+1}}^{t+1}$ has cardinality equal to $k$. 
The optimality of $\bm{z}^{t+1}_{\Theta^{t+1}}$ implies $\|\bm{z}^{t+1}_{\Theta^{t+1}}-\bm{d}_{\Theta^{t+1}}^{t+1}\|_2\le \|\bm{x}_{\Theta^{t+1}}-\bm{d}_{\Theta^{t+1}}^{t+1}\|_2$. Plugging the latter inequality back into \eqref{eq:triangle} yields
\begin{equation}
\big\|\bm{x}_{\Theta^{t+1}}-\bm{z}^{t+1}_{\Theta^{t+1}}\big\|_2\le 2\big\|\bm{x}_{\Theta^{t+1}}-\bm{d}^{t+1}_{\Theta^{t+1}}\big\|_2
\label{eq:2times}.
\end{equation}

Define the estimation error $\bm{h}^{t}:=\bm{x}-\bm{z}^{t}$.
Rewriting and substituting 
\begin{align*}
\bm{d}^{t+1}&=\bm{z}^{t}-\frac{\mu}{m}\sum_{i=1}^m\big(\bm{a}_i^\ccalT\bm{z}^{t}-\bm{a}_i^\ccalT\bm{x}\big)
\bm{a}_i+\frac{\mu}{m}\sum_{i=1}^m\Big(\frac{\bm{a}_i^\ccalT\bm{z}^{t}}{|\bm{a}_i^\ccalT\bm{z}^{t}|}-\frac{\bm{a}_i^\ccalT\bm{x}}{|\bm{a}_i^\ccalT\bm{x}|}\Big)|\bm{a}_i^\ccalT\bm{x}|\bm{a}_i
\end{align*}
 into \eqref{eq:2times} leads to
\begin{align}
&\frac{1}{2}\|\bm{h}^{t+1}_{\Theta^{t+1}}\|_2
	\le \Big\|\bm{h}_{\Theta^{t+1}}^t-\frac{\mu}{m}\sum_{i=1}^m\bm{a}_i^\ccalT\bm{h}^t\bm{a}_{i,\Theta^{t+1}}
-\frac{\mu}{m}\sum_{i=1}^m\Big(\frac{\bm{a}_i^\ccalT\bm{z}^{t}}{|\bm{a}_i^\ccalT\bm{z}^{t}|}-\frac{\bm{a}_i^\ccalT\bm{x}}{|\bm{a}_i^\ccalT\bm{x}|}\Big)|\bm{a}_i^\ccalT\bm{x}|\bm{a}_{i,\Theta^{t+1}}\Big\|_2\nonumber\\
	&=\Big\|\bm{h}_{\Theta^{t+1}}^t-\frac{\mu}{m}\sum_{i=1}^m\bm{a}_{i,\Theta^{t+1}}\bm{a}_{i,\Theta^{t+1}}^\ccalT\bm{h}_{\Theta^{t+1}}^t-\frac{\mu}{m}\sum_{i=1}^m\bm{a}_{i,\Theta^{t+1}}\bm{a}_{i,\Theta^{t}\setminus\Theta^{t+1}}^\ccalT\bm{h}_{\Theta^{t}\setminus\Theta^{t+1}}^t
	\nonumber\\
	& \quad-\frac{\mu}{m}\sum_{i=1}^m\Big(\frac{\bm{a}_i^\ccalT\bm{z}^{t}}{|\bm{a}_i^\ccalT\bm{z}^{t}|}-\frac{\bm{a}_i^\ccalT\bm{x}}{|\bm{a}_i^\ccalT\bm{x}|}\Big)|\bm{a}_i^\ccalT\bm{x}|\bm{a}_{i,\Theta^{t+1}}\Big\|_2\nonumber\\
	&\le \Big\|\bm{h}_{\Theta^{t+1}}^t-\frac{\mu}{m}\sum_{i=1}^m\bm{a}_{i,\Theta^{t+1}}\bm{a}_{i,\Theta^{t+1}}^\ccalT\bm{h}_{\Theta^{t+1}}^t\Big\|_2+\Big\|\frac{\mu}{m}\sum_{i=1}^m\bm{a}_{i,\Theta^{t+1}}\bm{a}_{i,\Theta^{t}\setminus\Theta^{t+1}}^\ccalT\bm{h}_{\Theta^{t}\setminus\Theta^{t+1}}^t\Big\|_2\nonumber\\
	& \quad+\Big\|\frac{\mu}{m}\sum_{i=1}^m\!\Big(\frac{\bm{a}_i^\ccalT\bm{z}^{t}}{|\bm{a}_i^\ccalT\bm{z}^{t}|}\!-\!\frac{\bm{a}_i^\ccalT\bm{x}}{|\bm{a}_i^\ccalT\bm{x}|}\Big)|\bm{a}_i^\ccalT\bm{x}|\bm{a}_{i,\Theta^{t+1}}\Big\|_2\label{eq:expand}
\end{align}
where the equality follows from re-expressing
$\bm{a}_i^\ccalT\bm{h}^{t}=\bm{a}_{i,\Theta^{t}}^\ccalT\bm{h}^{t}_{\Theta^{t}}=
\bm{a}_{i,\Theta^{t+1}}^\ccalT\bm{h}^t_{\Theta^{t+1}}+\bm{a}_{i,\Theta^t\setminus\Theta^{t+1}}^\ccalT\bm{h}^{t}_{\Theta^{t}\setminus\Theta^{t+1}}$ since $\bm{h}^{t}=\bm{x}-\bm{z}^{t}$ is supported on $\Theta^{t}$. The last inequality is readily obtained with triangle inequality of the $\ell_2$-norm.

The task now remains to establish upper bounds for the three terms appearing on the right hand side of \eqref{eq:expand}, which will be the subject for the rest of this section. 
Toward this end, let us recall the concept of the so-called restricted isometry property (RIP) condition in compressive sampling \cite{tit2005candes}. For each integer $ s=1,\,2,\,\ldots, k$, define the isometry constant $0<\delta_s<1$ of a matrix $\bm{\Phi}\in\mathbb{R}^{m\times n}$ as the smallest quantity such that the following holds for all $k$-sparse vectors $\bm{v}\in\mathbb{R}^n$ \cite{tit2005candes,acha2009nt}:
\begin{equation}
	(1-\delta_k)\|\bm{v}\|_2^2\le \|\bm{\Phi}\bm{v}\|_2^2\le (1+\delta_k)\|\bm{v}\|_2^2\label{eq:rip}.
\end{equation}

For Gaussian matrix $\bm{A}\in\mathbb{R}^{m\times n}$ whose entries are i.i.d. standard normal variables, then $\frac{1}{\sqrt{m}}\bm{A}$ satisfies the RIP with constant $\delta_{3k}\le \epsilon$ with probability at least $1-{\rm e}^{-c_0'm}$, provided that $m\ge C_1'\epsilon^{-2}(3k)\log(n/(3k))$ for certain universal constants $c_0',\,C_1'>0$ \cite{tit2005candes},  \cite[Eq. (1.2)]{acha2009nt}. 
Furthermore, if $\mathcal{K}\subsetneqq\{1,\,2,\,\ldots,\,n\}$ is a set of $3k$ indices or fewer, the following properties of $\bm{A}$ hold true \cite[Prop. 3.1]{acha2009nt}:
\begin{enumerate}
	\item[\textbf{P1)}] $\|\bm{A}_{\mathcal{K}}^\mathcal{T}\bm{u}\|_2\le \sqrt{(1+\delta_{3k})m}\|\bm{u}\|_2$, for all $\bm{u}\in\mathbb{R}^m$;
	\item[\textbf{P2)}] $(1-\delta_{3k})m\|\bm{v}\|_2\le \|\bm{A}_{\mathcal{K}}^\ccalT\bm{A}_{\mathcal{K}}\bm{v}\|_2\le (1+\delta_{3k})m\|\bm{v}\|_2$, for all at most $3k$-sparse vectors $\bm{v}\in\mathbb{R}^n$;
	\item[\textbf{P3)}] $\|\bm{A}_{\mathcal{B}}^\ccalT\bm{A}_{\mathcal{D}}\|_2\le \delta_{3k}$, where $\mathcal{B}$ and $\mathcal{D}$ are disjoint sets of combined cardinality not exceeding $3k$;
	\item[\textbf{P4)}] $\|\bm{A}_{\mathcal{B}\cup \mathcal{D}}^\ccalT\bm{A}_{\mathcal{B}\cup \mathcal{D}}-\bm{I}\|_2\le\delta_{3k} $.
\end{enumerate}
   
Having elaborated on the properties of RIP matrices, we are ready to derive bounds for the three terms on the right hand side of \eqref{eq:expand}. 
Regarding the first term, it is easy to check that
 \begin{align}
&\quad	\Big\|\bm{h}_{\Theta^{t+1}}^t\!-\!\frac{\mu}{m}\sum_{i=1}^m\bm{a}_{i,\Theta^{t+1}}\bm{a}_{i,\Theta^{t+1}}^\ccalT\bm{h}_{\Theta^{t+1}}^t\Big\|_2\nonumber\\
	&=\Big\|\Big(\bm{I}-\!\frac{\mu}{m}\sum_{i=1}^m\bm{a}_{i,\Theta^{t+1}}\bm{a}_{i,\Theta^{t+1}}^\ccalT\Big)\bm{h}_{\Theta^{t+1}}^t\Big\|_2\nonumber\\
	&\le\Big\|\bm{I}-\frac{\mu}{m}\sum_{i=1}^m\bm{a}_{i,\Theta^{t+1}}\bm{a}_{i,\Theta^{t+1}}^\ccalT\Big\|_2\big\|\bm{h}_{\Theta^{t+1}}^t\big\|_2\nonumber\\
	&\le \max\!\big\{1-\mu\underline{\lambda},\,\mu\widebar{\lambda}-1\big\}\big\|\bm{h}_{\Theta^{t+1}}^t\big\|_2\label{eq:1stterm}
\end{align}
where $\widebar{\lambda},\,\underline{\lambda}>0$ are the largest and smallest eigenvalue of $(1/m)\sum_{i=1}^m\bm{a}_{i,\Theta^{t+1}}\bm{a}_{i,\Theta^{t+1}}^\ccalT$, respectively. Specifically, the two inequalities in \eqref{eq:1stterm} are obtained based on the definition of the induced $2$-norm (i.e., the spectral norm) of matrices. 

Next, we estimate the eigenvalues $\widebar{\lambda}$ and $\underline{\lambda}$. Using P2, it clearly holds that
\begin{align}\label{eq:lambdamax}
\widebar{\lambda}&= \lambda_{\max}\Big(\frac{1}{m}\sum_{i=1}^m\bm{a}_{i,\Theta^{t+1}}\bm{a}_{i,\Theta^{t+1}}^\ccalT\Big)
\le 1+\delta_{2k}
\end{align}
due to $|\Theta^{t+1}|\le2k$. 
For the same reason, it further holds that
 \begin{equation}\label{eq:lambdamin}
\underline{\lambda}= \lambda_{\min}\Big(\frac{1}{m}\sum_{i=1}^m\bm{a}_{i,\Theta^{t+1}}\bm{a}_{i,\Theta^{t+1}}^\ccalT\Big)\ge 1-\delta_{2k}.
\end{equation}
Taking the results in \eqref{eq:lambdamax} and \eqref{eq:lambdamin} into \eqref{eq:1stterm} yields
\begin{align}\label{eq:1sttermfinal}
&\quad	\Big\|\bm{h}_{\Theta^{t+1}}^t\!-\frac{\mu}{m}\sum_{i=1}^m\bm{a}_{i,\Theta^{t+1}}\bm{a}_{i,\Theta^{t+1}}^\ccalT\bm{h}_{\Theta^{t+1}}^t
\Big\|_2\nonumber\\
&	\le \max\!\big\{1-\mu(1-\delta_{2k}),\,\mu(1+\delta_{2k})-1\big\}\big\|\bm{h}_{\Theta^{t+1}}^t\big\|_2.
\end{align}

For the second term in \eqref{eq:expand}, since $|\Theta^{t+1}\cup\Theta^{t}|\le 3k$, the next holds with high probability
	\begin{align}
	\label{eq:2ndterm}
	&\Big\|\frac{1}{m}\sum_{i=1}^m\bm{a}_{i,\Theta^{t+1}}\bm{a}_{i,\Theta^{t}\setminus\Theta^{t+1}}^\ccalT\bm{h}_{\Theta^{t}\setminus\Theta^{t+1}}^t\Big\|_2
\nonumber\\
	&\le \Big\|\frac{1}{m}\sum_{i=1}^m\bm{a}_{i,\Theta^{t+1}}\bm{a}_{i,\Theta^{t}\setminus\Theta^{t+1}}^\ccalT\Big\|_2
\big\|\bm{h}_{\Theta^{t}\setminus\Theta^{t+1}}^t\big\|_2\nonumber\\
&\le\delta_{3k}\big\|\bm{h}_{\Theta^{t}\setminus\Theta^{t+1}}^t\big\|_2
\end{align}
in which the first inequality arises again from
 the definition of the matrix $2$-norm.
The last inequality can be obtained by appealing to P4. 

Consider now the last term in~\eqref{eq:expand}. For convenience, define $\bm{A}_{\Theta^{t+1}}^\ccalT:=[\bm{a}_{1,\Theta^{t+1}}~\cdots~\bm{a}_{m,\Theta^{t+1}}]$ with $|\Theta^{t+1}|\le 2k$, and also
$\bm{v}^t:=[v_1^t~\cdots~v_m^t]^\ccalT$ with $v_i^t:=(\frac{\bm{a}_i^\ccalT\bm{z}^{t}}{|\bm{a}_i^\ccalT\bm{z}^{t}|}-\frac{\bm{a}_i^\ccalT\bm{x}}{|\bm{a}_i^\ccalT\bm{x}|})|\bm{a}_i^\ccalT\bm{x}|$ for $i=1,\,\ldots,\,m$. Upon rearranging terms, the induced matrix $2$-norm definition
 implies that
\begin{align}\label{eq:3rdterm}
	\Big\|\frac{1}{m}\sum_{i=1}^m\Big(\frac{\bm{a}_i^\ccalT\bm{z}^{t}}{|\bm{a}_i^\ccalT\bm{z}^{t}|}-\frac{\bm{a}_i^\ccalT\bm{x}}{|\bm{a}_i^\ccalT\bm{x}|}\Big)|\bm{a}_i^\ccalT\bm{x}|\bm{a}_{i,\Theta^{t+1}}\Big\|_2
&=
	\frac{1}{m}\big\|\bm{A}_{\Theta^{t+1}}^\ccalT\bm{v}^t\big\|_2
	\nonumber\\
	&\le \Big\|\frac{1}{\sqrt{m}}\bm{A}_{\Theta^{t+1}}^\ccalT\Big\|_2\Big\|\frac{1}{\sqrt{m}}\bm{v}^t\Big\|_2.
\end{align}
Property P1 confirms that
 the largest singular value of $\bm{A}_{\Theta^{t+1}}^\ccalT\in\mathbb{R}^{m\times 2k}$ 
satisfies $s_{\max}(\bm{A}_{\Theta^{t+1}}^\ccalT)\le (1+\delta_{2k})\sqrt{m}$ with high probability.
Therefore, the following holds with high probability  
\begin{align}\label{eq:boundthethird}
\Big\|\frac{1}{m}\sum_{i=1}^m\Big(\frac{\bm{a}_i^\ccalT\bm{z}^{t}}{|\bm{a}_i^\ccalT\bm{z}^{t}|}-\frac{\bm{a}_i^\ccalT\bm{x}}{|\bm{a}_i^\ccalT\bm{x}|}\Big)|\bm{a}_i^\ccalT\bm{x}|\bm{a}_{i,\Theta^{t+1}}\Big\|_2	\le (1+\delta_{2k})\frac{1}{\sqrt{m}}\big\|\bm{v}^t\big\|_2.
\end{align}

For convenience, define the event
\begin{align}
\mathcal{K}_i&:=\left\{
\frac{\bm{a}_i^\ccalT\bm{z}}{|\bm{a}_i^\ccalT\bm{z}|}\ne\frac{ \bm{a}_i^\ccalT\bm{x}}{|\bm{a}_i^\ccalT\bm{x}|}\right\}.
\label{eq:kevent}
\end{align}
 Then, it follows that
\begin{align}
\label{eq:boundv}
\frac{1}{m}\left\|\bm{v}^t\right\|_2^2 
&=\frac{1}{m}\sum_{i=1}^m\Big(\frac{\bm{a}_i^\ccalT\bm{z}^{t}}{|\bm{a}_i^\ccalT\bm{z}^{t}|}-\frac{\bm{a}_i^\ccalT\bm{x}}{|\bm{a}_i^\ccalT\bm{x}|}\Big)^2|\bm{a}_i^\ccalT\bm{x}|^2\nonumber\\
&\le4\cdot\frac{1}{m}\sum_{i=1}^m |\bm{a}_i^\ccalT\bm{x}|\cdot|\bm{a}_i^\ccalT\bm{h}^t|\cdot\mathbb{1}_{ \mathcal{K}_i}\nonumber\\
&\le 
\frac{40}{9}\sqrt{1+\epsilon_1}\cdot\Big( \epsilon_1 +\frac{1}{10}\sqrt{\frac{21}{20}}\Big)
\big\|\bm{h}^t\big\|_2^2
\end{align}
where the first inequality follows upon substituting $|\bm{a}_i^\ccalT\bm{x}|\le |\bm{a}_i^\ccalT\bm{h}^t|$ on the event $\mathcal{K}_i$, and using $\big(\tfrac{\bm{a}_i^\ccalT\bm{z}^{t}}{|\bm{a}_i^\ccalT\bm{z}^{t}|}-\tfrac{\bm{a}_i^\ccalT\bm{x}}{|\bm{a}_i^\ccalT\bm{x}|}\big)^2\le 4$. The last inequality can be obtained by appealing to Lemma  \ref{le:smallprob} in the Appendix adapted from \cite[Lemma 7.17]{pwf}, which holds for all $(2k)$-sparse vectors $\bm{h}\in\mathbb{R}^n$. This result has also been employed in the recent sparse phase retrieval approach reported in  \cite{2017hedge}. Here, we set $\epsilon_0=1/10$ in \eqref{eq:smallprob}, and $\epsilon_1>0$ can take any sufficiently small values.

Plugging the inequality in \eqref{eq:boundv} into \eqref{eq:boundthethird} leads to
\begin{align}
\label{eq:3rdtermfinal}
&\quad	\Big\|\frac{1}{m}\sum_{i=1}^m\Big(\frac{\bm{a}_i^\ccalT\bm{z}^{t}}{|\bm{a}_i^\ccalT\bm{z}^{t}|}-\frac{\bm{a}_i^\ccalT\bm{x}}{|\bm{a}_i^\ccalT\bm{x}|}\Big)|\bm{a}_i^\ccalT\bm{x}|\bm{a}_{i,\Theta^{t+1}}\Big\|_2\nonumber\\
&\le (1+\delta_{2k})\cdot \sqrt{\frac{40}{9}}\sqrt{1+\epsilon_1}\cdot\Big( \epsilon_1 +\frac{1}{10}\sqrt{\frac{21}{20}}\Big)
\big\|\bm{h}^t\big\|_2\nonumber\\
&:=(1+\delta_{2k})
\zeta\big\|\bm{h}^t\big\|_2
\end{align}
where the constant is defined as
 $$\zeta:=\sqrt{\frac{40}{9}}\sqrt{1+\epsilon_1}\cdot\Big( \epsilon_1 +\frac{1}{10}\sqrt{\frac{21}{20}}\Big).$$
Substituting the three bounds in~\eqref{eq:1sttermfinal}, \eqref{eq:2ndterm}, and \eqref{eq:3rdtermfinal} into \eqref{eq:expand}, we obtain
\begin{align}
\big\|\bm{h}^{t+1}&\big\|_2\!\le 2\max\!\left\{1\!-\mu(1\!-\delta_{2k}),\,\mu(1+\delta_{2k})-\!1\right\}\!\big\|\bm{h}_{\Theta^{t+1}}^t\big\|_2
+2\mu\delta_{3k}\big\|\bm{h}_{\Theta^{t}\setminus\Theta^{t+1}}^t\big\|_2+2\mu(1+\delta_{2k})\zeta
\big\|\bm{h}^t\big\|_2\nonumber\\
&\le 2\sqrt{2}\max\!\big\{\!\max\!\left\{1-\mu(1\!-\delta_{2k}),\,\mu(1\!+\delta_{2k})-1\right\},\mu\delta_{3k}\big\}\|\bm{h}^t\|_2 +2\mu(1+\delta_{2k})\zeta
\big\|\bm{h}^t\big\|_2\nonumber\\
&\le 2\Big[ \sqrt{2}\max\big\{\!\max\!\left\{1-\!\mu(1\!-\delta_{2k}),\,\mu(1+\!\delta_{2k})-\!1\right\}, \mu\delta_{3k}
\big\} +\mu(1+\delta_{2k})\zeta
\Big]\big\|\bm{h}^t\big\|_2\nonumber\\
&:= \rho\big\|\bm{h}^t\big\|_2
\label{eq:finalbound}
\end{align}
where the second inequality follows from  $$\big\|\bm{h}_{\Theta^{t+1}}^t\big\|_2+\big\|\bm{h}^t_{\Theta^{t}\setminus\Theta^{t+1}}\big\|_2\le \sqrt{2}\,\big\|\bm{h}^t\big\|_2$$ over disjoint sets $\Theta^{t+1}$ and $\Theta^{t}\setminus\Theta^{t+1}$. To ensure linear convergence, it suffices to choose a constant step size $\mu>0$ such that 
\begin{align*}
\rho=2\big[ \sqrt{2}\max&\big\{\!\max\!\left\{1-\mu(1-\delta_{2k}),\,\mu(1+\delta_{2k})-1\right\}, \mu\delta_{3k}
\big\} +\mu(1+\delta_{2k})\zeta\big]<1.
\end{align*}

For sufficiently small $\delta_{3k}>0$ and $\epsilon_1>0$, one has $\nu:=1-\rho\in (0,1)$, which justifies the linear convergence result in \eqref{eq:thm}.   
\end{proof}

Theorem~\ref{thm:exact} can be directly implied by combining Lemmas~\ref{le: inverse},~\ref{le:dist}, and \ref{le:thresh}. In fact, Lemma~\ref{le: inverse} ensures exact support recovery so that the orthogonality-promoting initialization can be effectively performed on the equivalent dimension-reduced data samples. Lemma~\ref{le:dist} guarantees that the sparse initialization attained based on the dimensional-reduced data lands within a small neighborhood of the globally optimal solution (this region is also termed basin of attraction; see e.g.,~\cite{twf}, 
\cite{tit2016sl}, 
\cite{2016pkcs} for more details) with high probability. Starting from any point within the basin of attraction, Lemma~\ref{le:thresh} confirms that successive iterates of SPARTA will be dragged toward the globally optimal solution at a linear rate provided that the step size and the truncation threshold are appropriately selected.

\section{Concluding Remarks}
\label{sec:con}

This paper contributed a sparse truncated amplitude flow (SPARTA) algorithm for solving PR of sparse signals. SPARTA initially recovers the support of the underlying sparse signal, which is used to obtain a sparse orthogonality-promoting initialization using power iterations restricted on the estimated support; subsequently, SPARTA refines the initialization by means of hard thresholding based truncated gradient iterations  to ensure overall simplicity and scalability. SPARTA enjoys provably exact recovery as soon as the number of noiseless Gaussian measurements exceeds a certain bound. In contrast to state-of-the-art algorithms, such as AltMinPhase and TWF, SPARTA requires the same sample size but can afford lower computational complexity. Simulated tests corroborate markedly improved recovery performance and computational efficiency of SPARTA relative to existing alternatives. 

A few timely and pertinent extensions can be listed at this point. 
Instead of enforcing the $\ell_0$-pseudonorm constraint and the hard thresholding operation in SPARTA, it is worth investigating sparse PR by minimizing the empirical risk function \eqref{eq:cost} with convex or nonconvex sparsity-promoting regularization terms, e.g., the (reweighted) $\ell_1$-norm of the optimization variables. Developing stochastic optimization algorithms for both stages amenable to large-scale implementations is also pertinent. 
Generalizing SPARTA and our analytical results to robust sparse PR and matrix recovery with outliers constitute worthwhile future directions too \cite{mtwf,tsp2017lu,tsp2017chi}.

\section*{Appendix: Supporting Lemmas}


\begin{lemma}[\cite{jtp2003b}]
	\label{le:bounded}
	For i.i.d. zero-mean random variables $X_1,\,X_2,\,\ldots,\,X_m$, if there exists some nonrandom constant $b>0$ such that $X_i\le b$ for $1\le i\le m$, and $\mathbb{E}[X_i^2]=v^2$, then the following holds 
	\begin{equation}
		\label{eq:bounded}
		{\rm Pr}(X_1+\cdots+X_m\ge y)\le {\rm min}\Big(\exp\big(-\frac{y^2}{2\sigma^2}\big),c_0-c_0\Phi\big(\frac{y}{\sigma}\big)
		\Big)
	\end{equation}
	for $\sigma^2:=m\max(b^2,\,v^2)$, and the cumulative distribution function of the standard normal distribution $\Phi(\cdot)$, where one can take $c_0=25$.
\end{lemma}

\begin{lemma}[\cite{as2000lm}]
	\label{le:2000}
	Let $X_1,\,X_2,\,\ldots,\,X_m$ be i.i.d. Gaussian random variables with zero mean and variance $1$, and $b_1,\,b_2,\,\ldots,b_m$ be nonnegative.  
The following inequality holds for any $\epsilon>0$
	\begin{equation}
		\label{eq:2000}
{\rm Pr}\Big(\sum_{i=1}^m b_i(X_i^2-\!1)\!\ge\! 2\Big(\sum_{i=1}^mb_i^2\Big)^{\frac{1}{2}}\sqrt{\epsilon}+\!2\big(\!\max_{1\le i\le m}b_i\big)\epsilon\Big)\!\le\! \exp(-\epsilon).
	\end{equation}
\end{lemma}


\begin{lemma}~\cite[Lemma 7.17]{pwf}
	\label{le:smallprob}
		For any $k$-sparse $\bm{x}\in\mathbb{R}^n$ supported on $\mathcal{S}$, assume noise-free measurements $\psi_i=|\bm{a}_i^\ccalT\bm{x}|$ generated from i.i.d. Gaussian sampling vectors $\bm{a}_i\sim\mathcal{N}(\bm{0},\bm{I}_n)$, $i=1,\,2,\,\ldots,\,m$.
Fixing any $\epsilon_1>0$, and for all $(2k)$-sparse $\bm{h}\in\mathbb{R}^n$, 
 the following holds with probability at least $1-3{\rm e}^{-c_5 m}$
	\begin{align}
\frac{1}{m}\sum_{i=1}^m\Big(\frac{\bm{a}_i^\ccalT\bm{z}}{|\bm{a}_i^\ccalT\bm{z}|}-&\frac{\bm{a}_i^\ccalT\bm{x}}{|\bm{a}_i^\ccalT\bm{x}|}\Big)|\bm{a}_i^\ccalT\bm{x}|(\bm{a}_i^\ccalT\bm{h})\nonumber\\
&\le 2\frac{\sqrt{1+\epsilon_1}}{1-\rho_0}\Big(\epsilon_1+\sqrt{\frac{21}{20}}\rho_0\Big)	\|\bm{h}\|_2^2	\label{eq:smallprob}
\end{align}
for all $\bm{z}\in\mathbb{R}^n$ obeying $\|\bm{z}-\bm{x}\|_2\le \rho_0\|\bm{x}\|_2 $, provided that $m>c_6 (2s)\log(n/(2s))$ for some fixed numerical constants $c_5,\,c_6>0$. Here, $\rho_0=1/10$.
\end{lemma}

The proof of Lemma \ref{le:smallprob} can be found in \cite[Page 30]{pwf}, which generalizes the result of \cite[Lemma 3]{reshaped1}.

\section*{Acknowledgment}
The authors would like to thank the anonymous reviewers for their thorough review and all constructive comments and suggestions, which helped to improve the quality of the manuscript. 
The authors also thank Prof. Xiaodong Li for sharing the codes of the thresholded Wirtinger flow algorithm.

\small
\bibliographystyle{IEEEtran}
\bibliography{apower}

\end{document}